%% file: main.tex
\documentclass[a4paper,UKenglish,cleveref, autoref, thm-restate]{lipics-v2021}

\usepackage[T1]{fontenc}
\usepackage{graphicx}
\usepackage{color}

\usepackage{amsmath}
\usepackage{stmaryrd}
\usepackage{amssymb}
\usepackage{xspace}
\usepackage{ifthen}
\usepackage{cleveref}
\usepackage{tikz}
\usetikzlibrary{
shapes,
calc,
decorations.pathreplacing,
}
\usepackage[ruled]{algorithm}
\usepackage{algorithmic}

\input{common}

\ifthenelse{\isundefined{\extendedtrue}}{
    \newif\ifextended
    \extendedfalse
     \extendedtrue
}{
}

\title{Concurrent Stochastic Lossy Channel Games}

\titlerunning{CSLCG} 

\author{Daniel Stan}{EPITA, France \and \url{https://www.tudo.re/daniel.stan/}}{
    daniel.stan@epita.fr}{https://orcid.org/0000-0002-4723-5742}{}

\author{Muhammad Najib}{Heriot-Watt University, UK \and \url{https://valvestate.github.io/}}{m.najib@hw.ac.uk}{
    https://orcid.org/0000-0002-6289-5124
}{}

\author{Anthony Widjaja Lin}{University of Kaiserslautern-Landau, Max-Planck
Institute for Software Systems, Germany \and \url{https://anthonywlin.github.io/}}{awlin@mpi-sws.org}{
    https://orcid.org/0000-0003-4715-5096
}{supported by European Research Council under European Union’s Horizon
    research and innovation programme (grant agreement no 
\href{https://doi.org/10.3030/101089343}{501100000781})}

\author{Parosh Aziz Abdulla}{Uppsala University, Sweden \and \url{https://user.it.uu.se/~parosh/}}{parosh@it.uu.se}{
    https://orcid.org/0000-0001-6832-6611
}{}

\authorrunning{D. Stan, M. Najib, A.\,W. Lin, and P.\,A. Abdulla} 

\Copyright{Daniel Stan, Muhammad Najib, Anthony Widjaja Lin, and Parosh Aziz Abdulla} 



\category{} 

\ifextended
\else
\relatedversiondetails[cite={bibtex-reference}]{Full Version}{https://arxiv.org/abs/...}
\fi

\nolinenumbers 

\ifextended
    \SeriesVolume{}
    \ArticleNo{-}
\else
\EventEditors{Aniello Murano and Alexandra Silva}
\EventNoEds{2}
\EventLongTitle{32nd EACSL Annual Conference on Computer Science Logic (CSL 2024)}
\EventShortTitle{CSL 2024}
\EventAcronym{CSL}
\EventYear{2024}
\EventDate{February 19--23, 2024}
\EventLocation{Naples, Italy}
\EventLogo{}
\SeriesVolume{288}
\ArticleNo{31}
\fi

\begin{document}

\maketitle

\begin{abstract}
Concurrent stochastic games are an important formalism for the \textit{rational
    verification} of probabilistic multi-agent systems, which involves verifying
    whether a temporal logic property is satisfied in some or all game-theoretic
    equilibria of such systems. In this work, we study the rational verification of probabilistic multi-agent systems where agents can cooperate by communicating over \textit{unbounded lossy channels}. To model such systems, we present \textit{concurrent stochastic lossy channel games} (CSLCG) and employ an equilibrium concept from cooperative game theory known as the \textit{core}, which is the most fundamental and widely studied cooperative equilibrium concept. Our main contribution is twofold. First, we show that the rational verification problem is undecidable for systems whose agents have almost-sure LTL objectives. Second, we provide a decidable fragment of such a class of objectives that subsumes almost-sure reachability and safety. Our techniques involve reductions to solving infinite-state zero-sum games with conjunctions of qualitative objectives. To the best of our knowledge, our result represents the first decidability result on the rational verification of stochastic multi-agent systems on infinite arenas.
\end{abstract}


\input{introduction}

\input{prelim}

\input{clcg}

\input{zerosum}

\input{conjobj}

\input{core}

\input{conclusion}



\bibliography{bib}

\ifextended
  \appendix
  \input{appendix}
\fi

\end{document}

%% file: common.tex
\usepackage{ifthen}
\ifthenelse{\isundefined{\drafttrue}}{
    \newif\ifdraft
    \draftfalse
}{}

\ifdraft
    \usepackage{showlabels}
\fi
\newcommand{\defn}[1]{\emph{#1}}
\newcommand{\Act}{\mathrm{Act}}
\newcommand{\Agt}{\mathrm{Agt}}
\newcommand{\Tab}{\mathrm{Tab}}
\newcommand{\op}{\mathrm{op}}
\newcommand{\dist}{\mathrm{Dist}}
\newcommand{\calG}{\mathcal{G}}
\newcommand{\calA}{\mathcal{A}}
\newcommand{\calF}{\mathcal{F}}
\newcommand{\Arena}{\mathcal{A}}

\newcommand{\map}{\eta}
\newcommand{\pre}{\mathrm{Pre}}
\newcommand{\stay}{\mathrm{Stay}}
\newcommand{\strat}{\mathfrak{S}}

\newcommand{\uniform}{\mathcal{U}}

\newcommand{\cleq}{\sqsubseteq}
\renewcommand{\Pr}{\mathbb{P}}
\newcommand{\esp}{\mathbb{E}}
\newcommand{\cyl}{\mathrm{Cyl}}

\newcommand{\F}{\Diamond}
\newcommand{\G}{\Box}

\newcommand{\X}{{\bigcirc}}

\newcommand{\bbN}{\mathbb{N}}
\newcommand{\bbRp}{\mathbb{R}_{\geq 0}}

\newcommand{\AS}{\mathrm{AS}}
\newcommand{\NZ}{\mathrm{NZ}}

\newcommand{\sem}[1]{\llbracket #1 \rrbracket}

\renewcommand{\Game}{\mathcal{G}}
\newcommand{\WinSet}{\mathrm{W}}
\newcommand{\LoseSet}{\mathrm{L}}
\newcommand{\Core}{\mathrm{Core}}
\newcommand{\achan}{}

\newcommand{\ecore}{\textsc{E-Core}\xspace}
\newcommand{\acore}{\textsc{A-Core}\xspace}

\newcommand{\cconf}{\mu}

\ifdraft
\newcommand{\todo}[1]{ \textcolor{red}{\textbf{TODO}: #1}}
\else
\newcommand{\todo}[1]{}
\fi

\newcommand{\OMIT}[1]{}

\newcommand{\modelname}{CSLCG}

\ifdraft
\newcommand\anthony[1]{{\color{blue}
\tiny [#1 - \textbf{Anthony}]}}
\newcommand\daniel[1]{{\color{magenta}
\tiny [#1 - \textbf{Daniel}]}}
\newcommand\najib[1]{{\color{brown}
\tiny [#1 - \textbf{Najib}]}}
\newcommand\parosh[1]{{\color{cyan}
\tiny [#1 - \textbf{Parosh}]}}

\else
\newcommand\anthony[1]{}
\newcommand\daniel[1]{}
\newcommand\najib[1]{}
\newcommand\parosh[1]{}

\fi

%% file: introduction.tex
\section{Introduction}\label{sec:intro}
Rational verification concerns the problem of checking which temporal logic properties will be satisfied in game-theoretic equilibria of a multi-agent system, that is, the stable collective behaviours that arise assuming that agents choose strategies/policies rationally in order to achieve their goals~\cite{GHW17,abate2021rational}. The usual approach to rational verification is to model multi-agent systems as concurrent games~\cite{GHW17,GNPW20}. This involves converting a multi-agent system into a game where agents are represented by a collection of independent, self-interested players in a finite-state environment. The game is played over an infinite number of rounds, with each player/agent (we use these terms interchangeably throughout the paper) choosing an action to perform in each round. Each player's goal is typically given by a temporal logic formula, which the player aims to satisfy. The temporal logic formula may or may not be satisfied by the infinite plays generated from the game, assuming that the players act rationally to achieve their goals.

In this paper, unlike much of previous work in rational verification, we
consider systems that give rise to games with \textit{probabilistic transitions}
and \textit{infinitely many states}. In particular, we focus on systems that can
naturally be modelled by \textit{stochastic lossy channel games} \cite{AHAMS08}
in the \textit{multi-player} and \textit{concurrent} setting, where there are $n
\geq 2$ players who can make concurrent moves. Our setting generalises the
two-player turn-based framework presented in~\cite{AHAMS08}.
We call this model \textit{Concurrent Stochastic Lossy Channel Games} (CSLCG). 
This model can be used to analyse a wide class of systems that communicate through potentially unreliable FIFO channels, such as communication networks, timed protocols, distributed systems, and memory systems~\cite{DBLP:conf/lics/AbdullaAA16,DBLP:conf/lics/AbdullaJ93,DBLP:journals/tcs/AbdullaJ03,DBLP:journals/lmcs/AbdullaABN18}.
    Stochasticity can be used to represent uncertainty in both the 
    environment (e.g., branching and message losses) and the behaviour of  agents. 
    Incorporating such uncertainty is desirable from a practical standpoint, as real-world systems are expected to operate correctly even when communication is not perfect and agents' behaviour is not deterministic.
    In the context of memory systems, stochasticity is used as a
    fairness condition that prevents unrealistic scenarios where the shared
    memory is never updated by the processes~ \cite{AbdullaAAGK22,LahavNOPV21}.

Given the possibility of communication, albeit imperfect, among agents, it seems quite natural to assume that some form of cooperation may arise in games. Thus, a relevant and fundamental question within the rational verification framework is: \textit{``What temporal logic property is satisfied by the rational cooperation that emerges in such a setting?''}
To address this question, we consider an equilibrium concept from cooperative game theory called the \textit{core}~\cite{aumann61,scarf1967core,GKW19}, which is the most fundamental and widely-studied cooperative equilibrium concept. 
With this concept, the standard assumption is that there exists some mechanism\footnote{Such mechanism is assumed to be exogenous (e.g., via contracts) and beyond the scope of the present paper.} that the players in a game can use to make \textit{binding agreements}. These binding agreements enable players to cooperate and work in teams/coalitions, providing a way to eliminate undesirable equilibria that may arise in non-cooperative settings~\cite{GKW19,gutierrez2023cooperative}. We illustrate that this is also true in our setting in \Cref{ex:mitm}. Despite using a cooperative equilibrium concept, we emphasize that players are still self-interested, meaning they rationally pursue their individual goals. As such, the games we consider in this work are general-sum games instead of strictly positive-sum games, which are purely cooperative.

\subparagraph*{Contributions:}
We study the rational verification problem in CSLCG with the core as the key equilibrium concept. It is shown in \cite{GKW19} that the core of a game with \textit{qualitative objectives} can never be empty, which also applies to the setting considered in this paper. Thus, two relevant decision problems pertaining to rational verification in the present work are \ecore and \acore. \textsc{E-Core} asks whether there
\textit{exists} a strategy profile in the core satisfying a given property $
\Gamma $, whereas \textsc{A-Core} asks whether \textit{all} profiles in the core
satisfy $ \Gamma $. We first show that these problems are undecidable for games
in which the players' objectives and property $\Gamma$ are almost-sure LTL
formulae, i.e., of the form $\AS(\varphi)$, where $\varphi$ is a LTL formula. We consider LTL \textit{with regular valuations}~\cite{LTL-regular} where the set of \textit{states/configurations} satisfying an atomic proposition is represented by a regular language.
Then, for our main contribution, we show the following: when players' goals are given 
as \textit{almost-sure reachability} or \textit{almost-sure safety} objectives, and the property $ \Gamma $ is given as \textit{almost-sure reachability}, \textit{almost-sure safety}, or \textit{almost-sure B\"uchi},
the problems of \textsc{E-Core} and 
\textsc{A-Core} become decidable. 
Our decidability proof is obtained via a reduction to \textit{concurrent} 2.5-player\footnote{Henceforth, we use the usual terms 1.5-player and 2.5-player games for, respectively, one-player and two-player stochastic games.} lossy channel games with conjunction of
objectives. This approach differs from previous work in two ways. First, our reduction considers \textit{concurrent} plays, in contrast to \textit{turn-based} 2.5-player games considered by \cite{AHAMS08}. Second, we do not assume finite-memory strategies, as opposed to finite-memory assumption in \cite{ACMS14,BBS07}. This is because finite-memory strategies do not ensure determinacy~\cite{Martin75}: It is possible that none of the players have a winning strategy in a given
concurrent 2.5-player game, even in the finite state case with simple
objectives~\cite{AHK07}. Therefore, general strategies (which may require infinite memory) are
required in equilibrium concepts such as the core, where players (or coalitions)
may try to satisfy their objectives while simultaneously preventing other players from achieving theirs. However, the main challenge in using these strategies in the infinite state case is the issue of representation. To address this, we provide a novel encoding of strategies in our proof of decidability.
To our knowledge, this is the first decidability result on 
the rational verification of stochastic multi-player games with infinite-state arenas.

\input{example2}

\subparagraph*{Related Work:}
As already mentioned above, the most relevant work w.r.t. verification of CSLCG is \cite{AHAMS08}, which
shows decidability of \textit{two-player turn-based} stochastic lossy channel games with
almost-sure reachability or almost-sure B\"{u}chi objectives. This work was
extended to parity conditions in \cite{ACMS14}, where decidability can be shown
\emph{assuming finite-memory strategies}; otherwise, it is already undecidable
for 1.5-player games over lossy channel systems with almost-sure co-B\"{u}chi
objectives \cite{BBS07}. 
%
We are not aware of any work on rational verification of concurrent stochastic games over infinite
arenas. \cite{HLNW21} studies verification of the core in a probabilistic setting, while \cite{aminof2019probabilistic} presents \textit{Probabilistic Strategy Logic}, which can be used to characterize the core. However, both of these works are in a finite state setting only. Without probability, we mention the work \cite{MP15,CSW16} on
concurrent (deterministic) pushdown games with multiple players. In particular, ATL* model
checking is decidable in such games, which allows one to reason about the
core. Additionally, there has been work on pushdown \textit{module-checking}, which provides some element of non-determinism through an (external) environment. \cite{aminof2013pushdown} examines the imperfect information setting, while \cite{bozzelli2020module} studies multi-agent systems with ATL* specifications. Note that lossy channel systems, which are the focus of our work, are inherently different from the models considered in these studies.

\subparagraph*{Organization:}
\Cref{sec:prelim} introduces preliminary definitions and notations.
\Cref{sec:clcg} describes
concurrent lossy channel games and the special case of 2.5-player zero-sum games.
\Cref{sec:core-veri} presents a characterization of the core, the problems \ecore and \acore, a procedure to solve them, and an undecidability result of \ecore and \acore.
\Cref{sec:zerosum} addresses the computability of winning regions for concurrent 2.5-player zero-sum games and provides
algorithms to compute such regions.
\Cref{sec:conjobj} studies the conjunction of objectives, while \cref{sec:core}
presents our main result on the decidability of \ecore and \acore. Finally, \Cref{sec:conclusion} concludes with a discussion and future work.

%% file: example2.tex
\newcommand{\tower}[1]{
  \coordinate (b1) at (#1);
  \coordinate (b2) at ($(b1)+(1,0)$);
  \coordinate (base) at ($(b1)!0.5!(b2)$);
  \coordinate (top) at ($(base)+(0,+1.4)$);
  \coordinate (t1) at ($(top)+(-0.25,0)$);
  \coordinate (t2) at ($(top)+(+0.25,0)$);

  \node[thick, above of=top, inner sep=2, draw, circle, node distance=10] (ant) {};
  \draw[thick] (b1) -- (t1) -- (t2) -- (b2)
        (top) -- (ant);
  
  \coordinate (l0) at ($(b1)!0.1!(t1)$);
  \coordinate (r0) at ($(b2)!0.1!(t2)$);
  \coordinate (l1) at ($(b1)!0.4!(t1)$);
  \coordinate (r1) at ($(b2)!0.4!(t2)$);
  \coordinate (l2) at ($(b1)!0.7!(t1)$);
  \coordinate (r2) at ($(b2)!0.7!(t2)$);

  \draw (r0) -- (l1) -- (r2) -- (t1)
        (l0) -- (r1) -- (l2) -- (t2);

  \draw (ant) ++(45:.3) arc (45:-45:.3)
        (ant) ++(45:.5) arc (45:-45:.5)
        (ant) ++(40:.8) arc (40:-40:.8);
}

\newcommand{\channel}[1]{
\draw
    let
        \p1 = (begin),
        \p2 = (end),
        \p3 = (0,0.3)
    in
        ($(\x1,\y1+\y3)$) -- ($(\x2,\y2+\y3)$)
        ($(\x1,\y1-\y3)$) -- ($(\x2,\y2-\y3)$);

\coordinate (next) at (end);
\begin{scope}[every node/.style={draw, rectangle, node distance=0,
    execute at begin node={$\vphantom{b}$}}]
    \foreach \i in {#1}{
    \node[left of=next, anchor=east] (item) {$\i$};
    \coordinate (next) at (item.west);
    }
\end{scope}
}
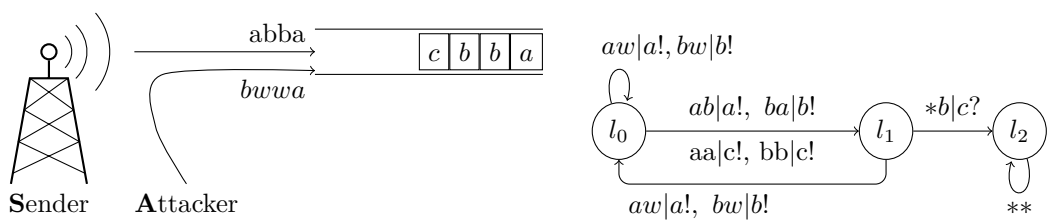
\begin{figure}
\begin{center}
\begin{tikzpicture}
\tower{$(0,0)$}
\node[below of=base, node distance=0,anchor=north] (sender)
    {\textbf{S}ender};

\coordinate (begin) at ($(ant)+(3.5,0)$);
\coordinate (end) at ($(begin)+(3,0)$);
\node (bla) at ($(ant)+(1,0)$) {};
\channel{a,b,b,c}

\coordinate (here) at (10,0);
\begin{scope}[xscale=-1,opacity=0.0]
\tower{here}
\node[below of=base, node distance=0,anchor=north]
    {\textbf{R}eceiver};
\end{scope}

\node[right of=sender, anchor=west] (attacker) {\textbf{A}ttacker};
\draw[->] (bla) edge
    node[pos=1, above,anchor=south east] {abba}
    (begin);
\draw[->] plot [smooth]
    coordinates {(attacker.north)
    ($(attacker)!0.5!(bla)$)
    ($(begin)+(-2,-0.30)$)
    ($(begin)+(0,-0.25)$)}
    node[below, anchor=north east] {$bwwa$}
    ;

        \begin{scope}[every node/.style={draw, ellipse}]
            \node (l0) at (8,0.7) {$l_0$};
            \node[right of=l0, node distance=10em] (l1) {$l_1$};
            \node[right of=l1, node distance=5em] (l2) {$l_2$};
        \end{scope}
        \draw[->] (l0) edge[loop above] node[above, anchor=south west,xshift=-1em] {$aw|a!,
                bw|b!$} (l0)
            edge
                node[above] {$ab|a!,~ba|b!$}
                node[below] {aa|c!,~bb|c!} 
                (l1)
            (l1) edge
                node[above] {$*b|c?$} (l2)
            (l2) edge[loop below] node[below] {$**$} (l2);
        \draw[->,rounded corners] (l1) -- ++(0,-2em)
            -- ++(-10em,0) node[below, anchor=north west] {$aw|a!,~bw|b!$} -- (l0);
\end{tikzpicture}
\end{center}
\caption{A simple adversarial transmission system from \cref{ex:2plexample} (left) and
    its graphical representation as a $2$-player arena (right). The resulting
    operation $f$ on the channel is written after $|$ and omitted if $f={nop}$.
}
\label{fig:example2}
\end{figure}
\begin{example}
    \label{ex:2plexample}
    To illustrate the model, we consider
    a simple transmission system depicted
    in \Cref{fig:example2} throughout the article. In this example,
    \defn{Sender} (player $1$) tries to emit some message of either type $a$ or $b$.
    \defn{Attacker} (player $2$) is trying to scramble the communication by
    concurrently choosing the same message type (action $a$ or $b$).
    Moreover, Attacker cannot scramble the communication two times in
    a row and has to wait (action $w$) otherwise.
    The \modelname{} arena is depicted on the right with the corresponding
    transitions and an extra location, reachable by a unilateral decision of
    player $1$ by reading a $c$-letter from the channel, which is possible only in case of a
    successful scrambling.
    Note that although the game structure is deterministic in this example, some stochastic
    behaviour will still appear both from message losses and from players'
    strategies, which are played concurrently.
    As a more concrete example of Attacker's objective, one could specify
    the condition
    ``reaching $l_2$ almost-surely, while not having more than 3 queued messages
    with positive probability''. Note that this is a conjunction of
    reachability and safety conditions over locations and regular sets of
    channel configurations.
\end{example}

%% file: prelim.tex
\section{Preliminaries}
    \label{sec:prelim}
     
    For a finite alphabet $\Sigma$, the set of finite sequences, called \defn{words},
    is written $\Sigma^*$.
    Given two words
        $u,v\in \Sigma^*$,
        we write $u\cdot v$ for their concatenation and extend this notation
        to sets
        of words. Given $L\subseteq \Sigma^*$, $L^+$ denotes the smallest set
        containing $L$ and closed under concatenation and
        $L^*=\{\epsilon\}\cup L^+$ with $\epsilon$ the empty word.
        The class of \defn{regular languages} is the smallest class containing
        $\Sigma$, closed under difference, union, Kleene star and concatenation.
        We refer to \cite{ullman01} for further references about regular
        expressions and their link to automata theory.

    Let $S$ denote a countable set, for example $S=\Sigma^*$.
    A \defn{well-quasi-ordering}~\cite{FS01} (wqo) $\preceq$ over $S$ is a quasi-ordering
    (i.e. reflexive and transitive binary relation)
    such that any infinite sequence $(s_i)_{i \in \bbN}$ of elements of $S$
    contains an increasing pair $i<j$ such that $s_i \preceq s_j$.
    As an example, Higman's lemma \cite{higman52} states that the
    \defn{sub-word ordering} $\preceq$ defined below is a wqo over
    $\Sigma^*$:
    \begin{definition}
        For any $w,w'\in \Sigma^*$, $w\cleq w'$ if
        $w$ can be written $w=w_0 \cdot w_1 \cdots w_n$ and
        $w'\in \Sigma^*\cdot \{w_0\} \cdot \Sigma^* \cdots
            \Sigma^* \cdot \{w_n\} \cdot \Sigma^*$.
    \end{definition}

    A subset
    $U\subseteq S$ is \defn{upward-closed} (UC) if for every
    $s\preceq t$ such that $s\in U$, we also have $t\in U$.
    Any UC set can be uniquely represented by its set of minimal elements,
    which is finite (wqo property). A \defn{downward-closed} (DC) set is defined
    in a similar manner.
    In particular, any UC or DC set w.r.t. the
    sub-word ordering is a regular language.

    A \defn{distribution over $S$} is an array $\delta\in \bbRp^{S}$
    of values $\delta[s]\in\bbRp$ for any $s\in S$, such that
    $\sum_{s\in S}\delta[s]=1$.
    If $X$ is finite and non-empty, we write $\uniform(X)$ for the
    \defn{uniform distribution over $X$}, namely
    $\forall x\in X, \uniform(X)[x] = 1/|X|$.
    We use the notation $\dist(S)$ to denote the set of distributions over $S$.

    Let $S^{\omega}$ denote the set of infinite sequences over $S$, called
    \defn{paths}.
    A set $X\subseteq S^{\omega}$ is a
    \defn{cylinder} if it is of the form
    $s_0 \cdots s_n \cdot S^{\omega}$ for some $s_0 \cdots s_n \in S^*$;
    Such a set is written $\cyl(s_0\cdots s_n)$.
    We introduce $\calF(S)$ as the smallest family of sets of paths containing
    all cylinders, closed under complementation, and countable union.
    Such sets of $\calF(S)$ are called \defn{measurable}.

    Given an initial state $s_0\in S$ and a mapping
    $\map: S^+ \rightarrow \dist(S)$
    from \defn{histories} to distributions over $S$, we define a partial function
    $\Pr: \calF(S) \rightarrow \bbRp$ such that
    $\Pr(\cyl(s_0))=1$,
    $\forall s\neq s_0, \Pr(\cyl(s))=0$,
    and for all $h\cdot s\in S^+\cdot S$,
    $\Pr(\cyl(h\cdot s))=\Pr(\cyl(h))\cdot\map(h)[s]$.
    Carathéodory's criterion~\cite{royden2010real} ensures this definition is
    well and
    uniquely defined on all $\calF(S)$, and $\Pr$ is therefore
    called a \defn{probability measure}.

    We describe infinite paths with the logic LTL,
    whose
    usual semantics over infinite words in $S^{\omega}$,
    leads to measurable sets~\cite[Remark~10.57]{BK08}.
    More precisely, we consider LTL with regular valuations~\cite{LTL-regular},
    where atomic propositions are represented by regular languages. We focus on the fragment without the ``until'' operator:
    \begin{definition}[LTL~\cite{Pnueli77}]
        \label{def:ltl}
        An LTL($\X,\F$) formula is any $\varphi$ in the following grammar,
        where $\nu$ ranges over regular languages over $S$:
        \[
            \varphi ::= 
                \nu
                ~|~\varphi\wedge\varphi
                ~|~\varphi\vee\varphi
                ~|~\X \varphi
                ~|~\F\varphi
                ~|~\G\varphi
        \]
    \end{definition}
    Here, $\G \varphi$, $\F \varphi$ and $\G \F \varphi$ denote
   \defn{always} $\varphi$, \defn{eventually} $\varphi$ and \defn{infinitely often} $\varphi$,
   respectively.

            Let $I$ be a set of indices and $V$ be a set of values.
            A \defn{profile} $\vec{v}$ is a mapping from $I$ to $V$, where
            $v_i$ is the value assigned to $i$. In particular if $I=\{1\hdots n\}$, then
            $\vec{v}=(v_1\hdots v_n)$.
            We introduce a fresh symbol $\bot\notin V$, and define the notation
            $\vec{v}_{-i}$ as the profile where $v_i$ has been
            replaced by $\bot$. Given any other value $w\in V$, we write
            $(\vec{v}_{-i},w)$ for the profile where the value assigned to $i$
            has been replaced by $w$. We extend these notations to a subset $ Y \subseteq I $ in the usual way, i.e., $ \vec{v}_{-Y} $ denotes $ \vec{v} $ where each $ v_i, i \in Y $ is replaced by $ \bot $ and $ (\vec{v}_{-Y},\vec{v}'_{Y}) $ where $ v_i $ is replaced by $ v_i' $ for each $ i \in Y $.


%% file: clcg.tex
\section{Lossy Channel Games}
    \label{sec:clcg}
    In this section, we provide formal definitions for the game model and
    the equilibrium concept that we use.
    While the model definitions are direct generalizations of those in~\cite{AHAMS08}\todo{cite others}, the
    concurrent setting requires extra care. In this setting, all players
    choose their actions concurrently and independently. The resulting
    action profile is evaluated on the game graph, which provides a
    (distribution of) channel operation to apply and a successor control state.
    Message losses are then processed.

    \subsection{Lossy Channel}
    For simplicity, this article focuses on a single
    channel system. The
    \defn{channel configuration} is represented by a word $\cconf\in M^*$, where $M$ is a finite alphabet of \defn{messages}.
    This channel is subject to stochastic
    \emph{message losses}, meaning that every message has a fixed
    probability $\lambda\in(0,1)$
    of being lost at every round, independently of other messages.

    We can derive the following probability values:
\begin{example}[Message Losses]
    \label{ex:messagelosses}
For any $\mu,\mu'\in M^*$, let us write
$P_{\lambda}({\mu,\mu'})$ for the probability of transitioning from channel configuration $\mu$ to
$\mu'$ after random message losses.
For example, $P_{\lambda}({\mu,\mu}) = (1-\lambda)^{|\mu|}$ (no message loss) and
whenever $\mu'\not\preceq \mu$ (not a sub-word), we have
$P_{\lambda}({\mu,\mu'})= 0$.
Moreover, for a single message letter $a\in M$,
$P_{\lambda}(a^n, a^m) = \binom{n}{m} \lambda^{n}(1-\lambda)^{n-m}$.
\end{example}

    As we will see later in~\Cref{sec:zerosum}, the exact value of $\lambda$ is
    not relevant for the qualitative probabilistic objectives considered
    in~\Cref{def:qualregions}.

    \subsection{Channel Operations}
    \begin{definition}
        \label{def:channelop}
        A \defn{channel operation} $f$ is defined as one of these
        three type of partial functions:
    \begin{itemize}
        \item If $f={nop}$ then $f(\cconf)=\cconf$;
        \item If $f=!m$ then $f(\cconf) = m\cdot \cconf$;
        \item If $f=?m$ and $\cconf=w\cdot m$,
            then $f(\cconf) = w$ and $f(\cconf) = \bot\notin M^*$ otherwise.
    \end{itemize}
    The set of all such partial functions is denoted $\op_{M}$.
\end{definition}
Intuitively, $nop$, $!m$, and
$?m$ denote ``no action'', ``enqueue the message $m$'',
and ``dequeue the message $m$'', respectively.
If the channel configuration
does not end with the message $m$, then the effect of $f=?m$ is the fresh
symbol $\bot$, indicating that the operation is not allowed.

\subsection{Lossy Channel Arena}
    Players choose channel operations through an \textit{arena}, which
    specifies a control state (\defn{location}) that defines the actions
    allowed by the players and the resulting effects on the channel:
\begin{definition}
    \label{def:clcg}
    A $n$-player
    \defn{concurrent stochastic lossy channel arena} (\modelname{} arena)
    is a tuple
    $\calA = (\Agt, L,M,\{\Act_i\}_{i \in \Agt},\Tab,l_0)$
    where:
    \begin{itemize}
    	\item $\Agt=\{1\hdots n\}$ is the set of agents;
    	\item $L$ is the set of locations;
    	\item $l_0\in L$ is the initial location;
    	\item $M$ is the message alphabet;
    	\item For each $ i \in\Agt $, $\Act_i$ is a finite set of actions available to agent $ i $ and require these sets to be pairwise disjoint. We write $ \Act $ for $ \prod_{i\in\Agt} \Act_i $; and
    	\item $ \Tab: L\times \Act \rightarrow
    	\dist(L\times \op_{M}) $.
    \end{itemize}
    A \defn{configuration}, or \defn{state} of a \modelname{} arena $ \calA $ is a word
    $s=l\cdot \mu$ composed of a location and a channel
    configuration.
    The \defn{state space} is denoted by $S=L \cdot M^*$ and
    the \defn{initial state} is $s_0 = l_0\cdot \epsilon \in S$.
\end{definition}
In the rest of the section we assume $\calA$ to be fixed.

\subsection{Concurrent Actions and Strategies}
    Since actions are taken concurrently, players must be prevented from
    taking certain actions that could result in illegal channel operations.
    In general, the set of allowed actions and strategy decisions
    will depend on the state (location \emph{and} channel)
    of the arena, as formalized below:
    \begin{definition}[Allowed Actions and Strategies]
    For a configuration $s=l\cdot \mu$ and a player $i\in\Agt$, we define
    $\Act_i(s)$ as the set of \defn{allowed actions} $\alpha$
    such that if
    $\exists \vec{\beta}\in \Act^n:
            \Tab(l, (\vec{\beta}_{-i},\alpha))[l',f] > 0$,
    then $f(\mu)\neq \bot$.
    A \defn{strategy for player $i$} is a mapping $\sigma_i$
    from sequences of states (\defn{histories}),
    to distributions of allowed actions. Namely, for all
        $h\cdot s \in S^+$, we have
        $\sigma_i(h\cdot s) \in \dist(\Act_i(s))$.
    We write $\sigma_i(\alpha~|~h\cdot s)$ as a shorthand for $\sigma_i(h\cdot s)[\alpha]$.
    The set of strategies of $i$ is written $\strat_i$ and
    $\vec{\strat}= \prod_{i \in \Agt} \strat_i$ is the set of \defn{strategy profiles}.
    \end{definition}
    
\begin{remark}
	\label{rem:jointdist}
	A strategy profile $\vec{\sigma}$
    can be seen as a strategy of a single player
	taking any (distribution of) action(s) $\vec{\alpha}\in \Act$ 
    from history $h\cdot s\in S^+$ with probability
	\[
	\vec{\sigma}(\vec{\alpha}~|~h\cdot s)=\prod_{i\in\Agt} \sigma_i(\alpha_i~|~h\cdot s)
	\]
	However the converse does not always hold:
	for example if $\Act_i(s)=\{a,b\}$ for $i\in \{1,2\}$,
	there exist no pair of strategies $\sigma_1,\sigma_2$ such that
	$\vec{\sigma}(aa~|~s) = \vec{\sigma}(bb~|~s)=1/2$ and
	$\vec{\sigma}(ab~|~s) = \vec{\sigma}(ba~|~s)=0$.
\end{remark}

\subsection{Semantics}
    \label{ssec:semantics}
    \input{ex2sem}
    The semantics of a \modelname{} arena can be understood as an $n$-player,
    infinite-state, concurrent stochastic game arena $(S, \Agt, \{\Act_i\}_{i \in \Agt},s_0, p)$ where
    $p: S \times \Act \rightarrow \dist(S)$\todo{cite}.
    Intuitively, from any state
    $s=l\cdot \mu$:
    (1) every player picks an allowed action from $\Act_i(s)$ based on the probability given by its strategy on the current history;
    (2) A new location $l'$ and a channel operation are then sampled
    according to $\Tab(l,\vec{\alpha})$. Since players only play allowed
    actions, the new channel configuration $f(\mu)$ is guaranteed to be defined; and (3) Message losses finally occur from $f(\mu)$ to some channel configuration
    $\mu'$, resulting in state $s'=l'\cdot \mu'$. This is formally defined below.
    An example of such a semantics is provided in \Cref{fig:statespace}.

\begin{definition}[Semantics of \modelname]
    \label{def:semclcg}
    For an initial state $s_0 \in L\cdot M^*$
    and a strategy profile
    $\vec{\sigma}=(\sigma_i)_{i\in\Agt}$ in
    the
    \modelname{} arena
    $\calA = (\Agt,L,M,\{\Act_i\}_{i \in \Agt},\Tab,l_0)$,
    we define $\Pr_{s_0}^{\vec{\sigma}}$ as the probability measure on
    $\calF(S)$ uniquely defined by its values on cylinders:
    \begin{itemize}
        \item $\Pr_{s_0}^{\vec{\sigma}}(\cyl(s_0)) = 1$;
        \item $\forall s\neq s_0, \Pr_{s_0}^{\vec{\sigma}}(\cyl(s)) = 0$;
        \item 
            For every pair of states
            $s=l\cdot \mu \in S$ and
            $s'=l'\cdot \mu' \in S$,
            and any history
            $h\in S^*$,
    \[
      \Pr_{s_0}^{\vec{\sigma}}(\cyl(h\cdot s \cdot s')) = 
      \Pr_{s_0}^{\vec{\sigma}}(\cyl(h\cdot s))
        \sum_{\vec{\alpha},f}
            \underbrace{
                \vec{\sigma}(\vec{\alpha}~|~h\cdot s)
            }_{(1)}
            \times
            \underbrace{
                \Tab(l,\vec{\alpha})[(l',f)]
            }_{(2)}
            \times
            \underbrace{
                P_{\lambda}(f(\mu),\mu')
            }_{(3)}.
    \]
    \end{itemize}
    When $s_0=l_0\cdot \epsilon$, we simply write
    $\Pr^{\vec{\sigma}}= 
     \Pr_{s_0}^{\vec{\sigma}} $.
\end{definition}

%

We now define qualitative objectives, that will be used to express
players' goals and the (global) properties to be checked:
\begin{definition}
	\label{def:qualregions}
	Let $\varphi$ be a measurable set of paths.
	For any strategy profile $\vec{\sigma}$ and state $s$, we define
    the \defn{property} $\NZ(\varphi)$ and $\AS(\varphi)$ by:
    \[
		\calA,
		\vec{\sigma},s\vDash
		\NZ(\varphi)
     \text{~for~}
		\Pr^{\vec{\sigma}}_s(\varphi) > 0
     \quad\text{and}\quad
		\calA,
		\vec{\sigma},s\vDash
		\AS(\varphi)
     \text{~for~}
		\Pr^{\vec{\sigma}}_s(\varphi) =1.
    \]
    We omit $\calA$ or $s$ when they are clear from context.

    We extend the definition to any conjunction of objectives: For a conjunction of $ \NZ $ and $ \AS $ objectives $ \Psi_1 \wedge \hdots \wedge \Psi_k $, we have
    \[
		\calA,
		\vec{\sigma},s\vDash
		\Psi_1 \wedge \hdots \wedge \Psi_k
     \text{~if, and only if, for all $i$~}
		\calA,
		\vec{\sigma},s\vDash
		\Psi_i.
    \]

\end{definition}
    
    We consider for $\varphi$ reachability or safety conditions of
    regular sets of states, namely LTL \textit{with regular valuations}~\cite{LTL-regular},
    and identify a formula $\varphi$ with its semantics
    $\mathcal{L}(\varphi)\subseteq (L\cdot M^*)^{\omega}$.
    For example, to specify the objective ``reaching $l_2$ almost-surely, while
    not having more than 3 queued messages with positive probability'', one would write:
    \[
        \Gamma = \AS(\F l_2\cdot M^*) \wedge 
        \NZ(\G L\cdot M^{\leq 3})
    \]

\begin{definition}
    \label{def:objectives}
    A \modelname{} is defined as an $n$-player arena together with
    such properties, called \defn{objectives}, for every players:
    $\calG=(\calA, \Phi_1 \hdots \Phi_n)$.
    In particular, when $n=2$, and $\Phi_1 = \neg\Phi_2$, we say that
    $\calG$ is a \defn{2.5 player zero-sum game}.
    A strategy $\sigma_i\in \strat_i$ such that
    $\forall \vec{\sigma}', (\vec{\sigma}'_{-i},\sigma_{i}), s\vDash \Phi$
    is called a \defn{winning strategy of $\Phi$ for $i$}.

\end{definition}

\begin{definition}\label{def:winlose}
	For a game $ \Game $, strategy profile $ \vec{\sigma} $, and state $ s $,
    we define the set of winners and losers by
    $ \WinSet_{\Game}(\vec{\sigma},s) = \{ i \in \Agt : (\vec{\sigma},s)
    \models \Phi_i \} $ and
    $ \LoseSet_{\Game}(\vec{\sigma},s) =
        \Agt \setminus \WinSet_{\Game}(\vec{\sigma},s) $.
    When $ s $ is the initial state we simply write
    $ \WinSet_{\Game}(\vec{\sigma}) $ and $ \LoseSet_{\Game}(\vec{\sigma}) $.
\end{definition}


\section{Verifying the Core}\label{sec:core-veri}
We consider a {cooperative} equilibrium concept called the {core}~\cite{aumann61,scarf1967core,GKW19}. 
Analogous to a Nash equilibrium (NE)~\cite{nash50}, a member of the core can be characterized by (the absence of) \textit{beneficial deviations}. However, unlike a Nash equilibrium where only one player can deviate, with the core, a \textit{group} or \textit{coalition} of players can deviate together. The notion of beneficial coalitional deviation is formally defined as follows.

\begin{definition}\label{def:ben-dev}
	For a strategy profile $ \vec{\sigma} $, we say that a joint strategy $ \vec{\sigma}_C', C \subseteq \Agt, C \neq \varnothing $, is a \textit{beneficial coalitional deviation} from $ \vec{\sigma} $ if $ C \subseteq \LoseSet_{\Game}(\vec{\sigma}) $ and for all $ \vec{\sigma}_{-C}' $, we have $ C \subseteq \WinSet_{\Game}((\vec{\sigma}_C',\vec{\sigma}_{-C}')) $.
\end{definition}

The core of a game $ \Game $ is defined to be \textit{the set of strategy profiles that admit no beneficial coalitional deviation}\footnote{An alert reader may notice the similarities between the core and \textit{strong NE}~\cite{aumann1959acceptable} and \textit{coalition-proof NE}~\cite{bernheim1987coalition}. The main difference is that these non-cooperative equilibrium concepts do not assume the existence of binding agreements. We refer to~\cite{gutierrez2023cooperative} for a more detailed discussion on this matter.}. We write $ \Core(\Game) $ to denote the set of strategy profiles in the core of $ \Game $.
%
We focus on two decision problems related to the core:
\textsc{E-Core} and \textsc{A-Core}. These problems are formally defined below.
\begin{quote}
	\emph{Given}: $(\Game, \Gamma)$ with game \(\Game\) and property \(\Gamma\). \\
	\emph{\ecore}: Does there exists some
	\(\vec{\sigma}\in\Core(\Game)\) such that \(\Game, \vec{\sigma}\models \Gamma\)?\\
	\emph{\acore}: Is it the case that for all \(\vec{\sigma}\in\Core(\Game)\) we have \(\Game, \vec{\sigma}\models \Gamma\)?
	%
\end{quote}
Note that, due to the duality of these problems, it is enough to provide a procedure for \textsc{E-Core}.


\input{example}

%
To address \ecore, we first introduce the notion of \textit{concurrent two-player coalition game} (TPCG) as follows. 

\begin{definition}\label{def:coal-game}
	Let $ \Game = (\Arena,(\Phi_i)_{i \in \Agt}) $ be  a \modelname{} and 
    $ C \subseteq \Agt $, with
    the underlying arena
	$ \Arena = (\Agt,L,M,\{\Act_i\}_{i \in \Agt},\Tab,l_0) $.
    The \textit{concurrent two-player coalition game arena} is defined as
    $ \Arena^C = ((1,2),L,M,\{\Act_i\}_{i \in \Agt},\Tab',l_0) $
    where for all actions $\vec{\alpha},\vec{\beta}\in \Act$,
    the transition is determined by the projections on $C$ and
    $-C = \Agt \setminus C$, respectively for player $1$ and $2$:
    $ \Tab'(l,(\vec{\alpha},\vec{\beta}))=
        \Tab(l,(\vec{\alpha}_C,\vec{\beta}_{-C}))$.
    The TPCG with respect to $ \Game, C, $ and objective $ \Psi_C $ is thus
    defined as $ \Game^{C,\Psi_C} = (\Arena^C,(\Psi_C,\neg \Psi_C)) $.
\end{definition}

Observe that the ability of a coalition $ C $ to satisfy an objective $ \Psi_C $ depends on whether it has a winning strategy in TPCG $ \Game^{C,\Psi_C} $ from the initial state of the game (thus, the existence of beneficial deviation).  With this observation, we restate the characterization of \textsc{E-Core} from \cite{HLNW21} as follows.

\begin{proposition}[\cite{HLNW21}]\label{prop:core-char}
	A pair $(\Game,\Gamma)$ is a yes-instance of \textsc{E-Core} if and only if there exists $ W \subseteq \Agt $ such that
	\begin{enumerate}
		\item[(a)] there exists some $ \vec{\sigma} $ such that $ \Game, \vec{\sigma} \models \Phi_{W} $, and
		\item[(b)] for all $ C \subseteq \Agt \setminus W $, $ C $ has no winning strategy in $ \Game^{C,\Psi_C} $,
	\end{enumerate}
	where $ \Phi_{W} = \bigwedge_{i\in W} \Phi_i \wedge \bigwedge_{i\notin W} \neg\Phi_i \wedge \Gamma  $ and $ \Psi_C = \bigwedge_{i \in C} \Phi_i $.
\end{proposition}

Using \cref{prop:core-char}, we provide the following procedure for determining whether some $ (\Game, \Gamma) $ is a yes-instance of \textsc{E-Core}.
\begin{enumerate}
	\item Guess a set of winning players $ W \subseteq \Agt $;
	\item Check if there is a winning strategy $ \vec{\sigma} $ in TPCG $ \Game^{\Agt,\Phi_{W}} = (\Arena, (\Phi_{W}, \neg \Phi_{W})) $;
	\item Check if there is a coalition $ C \subseteq \Agt \setminus W $ with a winning strategy $ \vec{\sigma}_{C}' $ in TPCG $ \Game^{C,\Psi_C} = (\Arena^C,(\Psi_C,\neg \Psi_C)) $;
	\item If the answers to Step 2 is ``Yes'' and Step 3 is ``No'', then $ (\Game, \Gamma) $ is a yes-instance of \textsc{E-Core}. Otherwise, it is not. 
\end{enumerate}

Observe that above procedure corresponds to \cref{prop:core-char}. In particular, Steps 2 and 3 correspond precisely to (a) and (b) in \cref{prop:core-char}, respectively. Thanks to this procedure, the problem of checking whether $ (\Game,\Gamma) $ is a yes-instance of \textsc{E-Core} can be reduced to solving a collection of concurrent $2.5$-player zero-sum games.

Note that if we consider almost-sure LTL objectives, we
can construct a CSLCG $ \Game $ with one player whose goal is
$ \Phi = \AS((\G\F R_1) \wedge (\F\G R_2)) $. Then, we can set
$ \Gamma = \Phi $ and query \ecore. This reduces to a 1.5-player game over a lossy channel with objective $ \Phi $, which is already
undecidable~\cite[Lemma~5.12]{BBS07}. Therefore, we obtain the following result.

\begin{proposition}
	\label{prop:undec}
	For a pair $ (\Game, \Gamma) $ where players' objectives and $ \Gamma $ are given by
    almost-sure LTL formulae, the problems of \ecore and \acore are undecidable.
\end{proposition}

In the following sections, we study the technical machineries required to solve
concurrent 2.5-player zero-sum games with almost-sure safety and almost-sure reachability
objectives. As motivated by \cref{prop:core-char} and its corresponding procedure for \ecore,
we further solve concurrent 2.5-player zero-sum games with a conjunction of objectives.
These provide the necessary foundation for the main decidability
result presented later in \cref{sec:core}.

%% file: ex2sem.tex
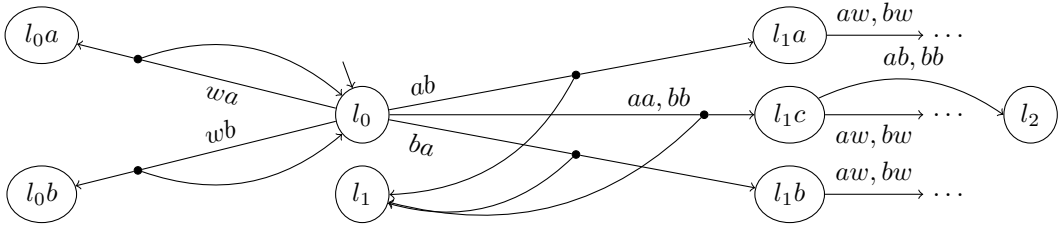
\begin{figure}
	\centering
    \begin{tikzpicture}
      \begin{scope}[every node/.style={draw, ellipse, node distance=30}]
        \node (l0) {$l_0$};
        \node[right of=l0, node distance=160] (l1c) {$l_1c$};
        \node[left of=l0, node distance=120, opacity=0] (l0ph) {};
        \node[above of=l1c] (l1a) {$l_1a$};
        \node[below of=l1c] (l1b) {$l_1b$};
        \node[below of=l0] (l1) {$l_1$};
        \node[above of=l0ph] (l0a) {$l_0a$};
        \node[below of=l0ph] (l0b) {$l_0b$};
      \end{scope}
      \begin{scope}[every node/.style={draw,ellipse, inner sep=1,fill=black}]
        \node (l0l1a) at ($(l0)!.5!(l1a)$) {};
        \node (l0l1b) at ($(l0)!.5!(l1b)$) {};
        \node (l0l1c) at ($(l0)!.8!(l1c)$) {};
        \node (l0l0a) at ($(l0)!.7!(l0a)$) {};
        \node (l0l0b) at ($(l0)!.7!(l0b)$) {};
      \end{scope}

      \draw (l0) edge[<-] +(110:0.75);
      \draw (l0) edge node[above, sloped,pos=0.2] {$ab$} (l0l1a)
                 edge node[above, yshift=-0.2em, pos=1, anchor=south east, sloped] {$aa,bb$} (l0l1c)
                 edge node[below, sloped, pos=0.2] {$ba$} (l0l1b)
                 edge node[below, sloped,pos=0.7,anchor=north west] {$wa$} (l0l0a)
                 edge node[above, sloped,pos=0.7,anchor=south west] {$wb$} (l0l0b);
      \draw[->] (l0l1a) edge (l1a)
                (l0l1c) edge (l1c)
                (l0l1b) edge (l1b)
                (l0l0a) edge (l0a)
                (l0l0b) edge (l0b);
      \draw[<-] (l1) edge[bend right] (l0l1a)
                     edge[bend right] (l0l1b)
                     edge[bend right] (l0l1c)
                (l0) edge[bend right] (l0l0a)
                     edge[bend left] (l0l0b);
      \draw[->] (l1a) edge
            node[above] {$aw,bw$} node[pos=1,anchor=west] {$\hdots$} +(5em,0);
      \draw[->] (l1b) edge
            node[above] {$aw,bw$} node[pos=1,anchor=west] {$\hdots$} +(5em,0);
      \draw[->] (l1c) edge
            node[below,sloped] {$aw,bw$} node[pos=1,anchor=west] {$\hdots$} +(5em,0em)
            edge[bend left] node[above,anchor=south] {$ab,bb$} 
                node[anchor=west,pos=1,draw,ellipse] {$l_2$} +(8em,0em);
                
    \end{tikzpicture}
	\caption{
        Fragment of the semantics of the \modelname{} arena of
        \cref{fig:example2} as a concurrent
        arena. Stochastic behaviours such as message losses are represented
        by the $\bullet$ nodes and probability values are omitted.
    }
	\label{fig:statespace}
\end{figure}

%% file: example.tex
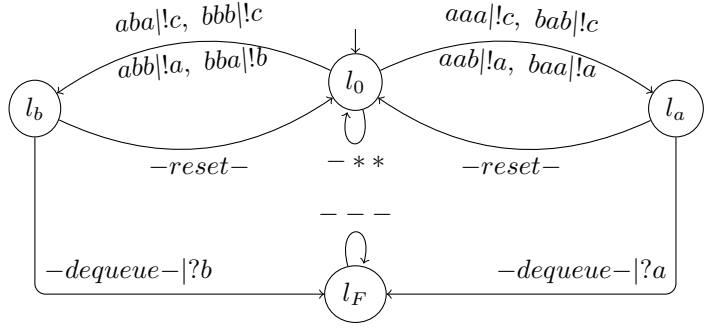
\begin{figure}
	\centering
    \begin{tikzpicture}
        \begin{scope}[every node/.style={draw, ellipse}]
            \node (l0) {$l_0$};
            \node (la) at (12em, -1em) {$l_a$};
            \node (lb) at (-12em, -1em) {$l_b$};
            \node (lf) at (0, -8em) {$l_F$};
        \end{scope}
        \draw[->,rounded corners] (l0) edge[<-] +(0,2em)
                       edge[bend left]
                       node[above, sloped]
                        {$aaa |\achan!c,~
                        bab|\achan!c$}
                       node[below, sloped]{$
                        aab|\achan!a,~
                        baa|\achan!a
                        $} (la)
                       edge[bend right] 
                       node[above, sloped]
                        {$aba | \achan!c,~
                        bbb|\achan!c$}
                       node[below, sloped]{$
                        abb|\achan!a,~
                        bba|\achan!b
                        $} (lb)
                      edge[loop below] node[below] {$-**$} (l0)
                   (lb)
                      edge[bend right]
                      node[below] {$-{reset}-$} (l0)
                      |-
                      node[below, anchor=south west] {$-{dequeue}{-}|\achan?b$} (lf);
        \draw[->,rounded corners]
                   (la)
                      edge[bend left]
                      node[below] {$-{reset}-$} (l0)
                      |-
                      node[below, anchor=south east] {$-{dequeue}{-}|\achan?a$} (lf)
                   (lf) edge[loop above] node[above] {$---$} (lf);
                   
%

    \end{tikzpicture}
	\caption{
    A graphical representation of a $3$-player arena. 
    Joint actions are ordered by $ SRA $, for example, $ xyz $ indicates that players $ S,R,A $ choose actions $ x,y,z$ respectively.
    }
	\label{fig:mitm}
\end{figure}

\begin{example}\label{ex:mitm}
	Consider a transmission system with one channel similar to the one discussed in \cref{ex:2plexample}. Now, suppose there are three players, $ S $ (sender), $ R $ (receiver), and
    $ A $ (attacker). 
    $ S $/$ R $ decides to send/request a message,
    $ a $ or $ b $. A message is delivered successfully if
    $ R $ reads (dequeues) the same type of message as she requested.
    An attack is successful if $ A $ chooses the same type of message being
    sent, and this turns the message into a $ c $. 
    We model this game as a
    \modelname\ where $ \Act_S = \Act_A =\{a,b,-\} $,
    $ \Act_R =\{a,b,{dequeue},reset\} $.
    The arena is depicted in \cref{fig:mitm}.
    The goals of $ S $ and $ R $ are given as
    $ \Phi_S = \AS(\G (L \cdot M^{\leq k})) $ for some $ k \in \bbN $ (almost surely the channel never exceeds
    size $ k $), and
    $ \Phi_R = \AS(\F(l_F \cdot M^*)) $ (almost surely a correct message is delivered).
    The goal of $ A $ is $ \Phi_A = \neg (\Phi_S \wedge \Phi_R) $.

	Consider the following strategy profile: if the channel contains fewer than $ k $ messages, $ S $ and $ R $
    play $ aa $ and $ bb $ uniformly at random. Otherwise, $ S $ plays action $ - $.
    This profile satisfies both $ \Phi_S $ and $ \Phi_R $, and is therefore in the core. In fact, \textit{all} strategy profiles in the core satisfy both $ \Phi_S $ and $ \Phi_R $. As a result, if we query \ecore with $ \Gamma = \Phi_A = \neg (\Phi_S \wedge \Phi_R) $, it returns a negative answer.
    This is in contrast to NE: since with NE $ S $ and $ R $ do not act as a coalition, there exists a (undesirable) NE in which there is a non-zero probability that the channel will exceed size $ k $ or a correct message will never be delivered (i.e., $ \Gamma $ is satisfied).
    $\blacksquare$
    \todo{Add explanation about NE and how it differs with the core}
\end{example}

%% file: zerosum.tex
\section{The Zero-Sum case is Effective}
\label{sec:zerosum}

In this section, we focus on solving concurrent 2.5-player zero-sum CSLCG as a crucial step towards our main decidability result for \ecore. We assume that the CSLCG $ \calG $ with arena $ \calA $ is fixed, and $ R $ represents a set of states, which may be infinite.
While \cite{AHK07} provides an approach for solving concurrent stochastic
games with $\AS(\F R)$, $\AS(\G R)$, $\NZ(\F R)$ and $\NZ(\G R)$
objectives in the \textit{finite-state} setting, here we extend their approach to the \textit{infinite-state} setting. We note that the infinite-state setting has been previously examined for \textit{turn-based} games in $\cite{BertrandS03,ABRS05,BaierBS06}$.

Our approach can be summarized as follows: First, we provide algorithms to
symbolically compute one-step
reachability for regular set of states. Then, we prove that the algorithms
from \cite{AHK07} remain valid (in terms of termination and correctness) for this
class.
More precisely, we fix a regular set $R\subseteq (L\cdot M^*)^*$ and
a 2.5-player zero-sum CSLCG $\calG$ with $\Agt = \{1,2\}$.
For $ i \in \Agt $, the objective $\Phi_i$ is of the form $\NZ(\F R)$
or $\AS(\F R)$.
Furthermore, by instantiating games for the opponent $-i$, setting
$R'=S\backslash R$, and by determinacy of such games~\cite{Martin75},
we also solve the game for objectives of the form $\AS(\G R')$ or $\NZ(\G R')$.
To this end, we focus on the existence of winning strategies (\cref{def:objectives}) and the computation of \textit{winning regions}.
\begin{definition}
    \label{def:winregion}
    For an objective $\Phi$ of player $ i $,
    the \defn{winning region of $\Phi$ for player $ i $} in $\calG$ is given
    by:
    $
        \sem{\Phi}_i(\calG) = \{s\in S~|~\exists \sigma_i\in \strat_i,
                \forall \vec{\sigma}'\in\vec{\strat},
            (\vec{\sigma}'_{-i},\sigma_{i}),s\vDash \Phi\}.
    $
\end{definition}

As previously mentioned, moving from a finite to an infinite state setting presents a challenge in the representation of winning strategies. To address this, in what follows, we provide a novel encoding of different classes of winning strategies necessary for ensuring a win.

\subsection{Regularity Properties of \modelname{}}
\label{ssec:regclcg}

We begin by revisiting the regularity results of lossy channel systems from
\cite{BertrandS03,ABRS05} to later design algorithms
for computing winning regions. In this context, the \textit{predecessor function} \todo{cite} and
the \textit{finite attractor}~\cite{BaierBS06} are essential notions that need to be adapted to the concurrent case.

\begin{definition}[Predecessor Function]
For a given player $i\in\Agt$
and a set of states $B\subseteq S$,
we write $\pre_i(B)$ the
set of states from which player $i$ can enforce reaching $B$ with
\textbf{positive probability}, no matter the actions of the other
players:
\[\pre_i(B) =
\{
s\in S~|~
\exists \sigma_i
\forall \vec{\sigma}'
\Pr^{(\vec{\sigma}'_{-i},\sigma_i)}_{s}(\X B)>0
\}
\]
\end{definition}

We summarize the properties below, which are all consequences of the lossy
nature of the system.
\begin{restatable}{proposition}{regclcg}
    \label{prop:regclcg}
    The order $\cleq$ defined for all $s=(l\cdot \mu),s'=(l'\cdot \mu')$ by
    $s\cleq s'$ if $s\preceq s'$ (sub-word) and
    $\mu,\mu'$ are equal or have the same last letter,
    is a wqo. Further, given $X\subseteq S$,
    \begin{enumerate}
        \item $\pre_i(X)$ is upward-closed for $\cleq$, hence regular;
        \item If $X$ is regular, then $\pre_i(X)$ can be computed;
        \item (\defn{Finite Attractor}~\cite{BaierBS06}) 
            There exists a finite set $A\subseteq S$,
            such that
            $\forall \vec{\sigma}\in\vec{\strat}~\Pr^{\vec{\sigma}}(\G\F A)=1$.
    \end{enumerate}
\end{restatable}

\begin{proof}[Proof sketch]
    \begin{enumerate}
        \item If $s \cleq t$, and $s \in \pre_i(X)$, then an action from
            $s$ leads to $X$ with positive probability. Since $s\cleq t$
            entails the same last letter in $s$ and $t$'s channel,
            $\Act_i(s) = \Act_i(t)$ and the same action
            can be played by $i$ from $t$: With positive probability
            the extra messages in $t$ and not in $s$ can then be dropped;
        \item We refer to the computability section of \cite{ACMS14}, and argue
            that the concurrent setting can be simulated
            by a turn-based game where the first player $i$ commits to an
            action $\alpha$, then moves to a new state where the rest of the
            players provide their action, and move to a stochastic state where
            messages are actually dropped. $\pre_i$ can therefore be simulated
            by three calls to the predecessor function of \cite{ACMS14}.
        \item We refer to Corollary~5.3 in \cite{BertrandS03} or
            Lemma~5.3 in \cite{ABRS05}. Intuitively,
            having more messages on the channel increases the likelihood of a decrease in
            message count, hence the set $L\cdot \epsilon$ (any location, empty channel)
            is shown to a finite attractor.
    \end{enumerate}
\end{proof}

\input{stratfig}
We provide the following regular encoding of strategies, which allows for a more precise description of classes of strategies and enables their finite representation:
\begin{definition}
    \label{def:fmem}
    Let $\sigma_i\in\strat_i$ be a strategy for player~$i$.
    $\sigma_i$ is \defn{positional} (P) if it depends only on the last
    state (location and channel):
    $\forall h,h'\in S^*,\forall s\in S,\sigma_i(hs)=\sigma_i(h's)$.

    $\sigma_i$ is \defn{finite memory} (FM)
    if
    \begin{itemize}
        \item The set $\{ \sigma_i(h) ~|~h\in S^+ \}\subseteq \dist(\Act_i)$
            is finite;
        \item For every $\delta\in \dist(\Act_i)$,
            the set of histories $\sigma^{-1}(\delta)\subseteq (L\cdot M^*)^*$
            is regular.
    \end{itemize}
\end{definition}

\cref{fig:strategies} illustrates how
a FM strategy can be represented finitely. It determines which
action to play by running a finite number of automata
on the history, reading every location and channel configuration.
In contrast, a strategy that depends only on the last state may not
be representable in the infinite state case, as infinitely many different
decisions might be taken. Therefore, we avoid using the usually synonymous term
\defn{memoryless} and instead refer to these strategies as \defn{positional}.
A convenient representation is therefore the
positional
\emph{and} finite memory (PFM) strategies, meaning that the action
depends only on the last state, which must belong to one of finitely many
regular sets.

\subsection{Positive Reachability}
\label{ssec:posreach}
In the positive reachability case, player~$i$ tries to enforce that \emph{some}
finite prefix reaches the target set $R$. This can be achieved by a
backward-reachability algorithm instantiated on well-quasi-ordered
sets~\cite{ACJT96}:


\begin{restatable}{lemma}{lemreachalgo}
    \label{lem:reachalgo}
    Given a regular set $R\subseteq S^*$,
    the algorithm that computes
    $\bigcup_{k\geq 0} \pre^k_i(R)$
    converges in a finite number of steps and returns
    a set $R\cup V$, where $V$ is upward-closed and
    $
        R\cup V = 
            \sem{ \NZ( \F R) }_i
    $.
    Moreover, PFM winning strategies are sufficient
    for both players.
\end{restatable}

\begin{example}
    \label{ex:posreach}
    Consider the example from \Cref{fig:example2}. A winning strategy for
    player 2/Attacker to achieve the objective $\NZ(\F l_2\cdot c)$ (eventually
    reaching state $l_2\cdot c$ with positive probability) consists of
    playing $\uniform(\{a,b\})$ from any state in $l_0\cdot M^*$,
    then $b$ from any state in $l_1\cdot M^*\cdot c^2$.
    With positive probability, the
    state $l_1\cdot cc$ is reached after $3$ steps, regardless of Sender's
    strategy, and then $l_2\cdot c$ is reached.
\end{example}

This algorithm can also be used to compute almost-sure safety of $R$
by applying the determinacy result:
$\sem{\AS(\G R)}_i = \overline{\sem{\NZ(\F \overline{R})}_{-i}}$.
As shown in~\cref{lem:reachalgo}, the optimal strategy is PFM, and in the safety case, it can be further
restricted without compromising the safety property. More precisely,
we define the most general action restriction for player~$i$ that preserves
safety as follows:
\begin{definition}
    \label{def:stay}
For any set $R\subseteq S$,
$\stay_i^\calG(R)$ is the mapping that restricts the allowed actions
on $R$ to stay in $R$:
\[
    s\in R\mapsto
              \{\alpha\in \Act_i(s)~|~
                \forall \vec{\sigma},
                \Pr^{(\vec{\sigma}_{-i},\alpha)}(\X R)=1
              \}
\]
\end{definition}

\subsection{Almost-Sure Reachability}
\label{ssec:asrreach}
In contrast to the turn-based case in~\cite{AHAMS08},
concurrent actions require a careful analysis of the allowed actions.
The intuition is as follows: As the player with reachability objective tries to avoid ``being trapped'' in bad states,
he may choose not to play some actions
based on the previously defined $\stay$ operator (\cref{def:stay}).
This limits the available actions and gives more power to his opponent, which then reduces the winning region.
This idea was proposed in the algorithm for almost-sure reachability in finite
state games described in \cite{AHK07}.
We follow here this approach by implementing
the algorithm symbolically as depicted in \cref{alg:asreach} which boils down to
providing effective procedures for $\pre_i$ and $\stay_i$, when the input
is a regular set of states.
As the algorithm in \cite{AHK07} actually solves almost-sure \textit{repeated} reachability, the authors
assumed that the target set $R$ is \defn{absorbing},
meaning that once a state in $R$ is reached, the game cannot exit $R$, regardless
of the players' actions. This assumption is also made here as adjusting
the $\pre_i$ operator is sufficient to guarantee this property in our setting.
Termination of the algorithm can be shown using the wqo property,
similar to that described in~\cite{AHAMS08}, resulting in the following lemma:

\begin{restatable}[t]{algorithm}{asreach}
    \caption{Almost-Sure Reachability}
    \label{alg:asreach}
    \textbf{Input}: An arena $\calA$, $i\in\Agt$ and a regular set $R$\\
    \textbf{Output}: $D_k=\sem{\AS(\G \F R)}_i$
    \begin{algorithmic}
        \STATE $k\gets 0;\quad D_0 \gets S; \quad \calA_0 \gets \calA$
        \REPEAT
            \STATE $C_{k} =
                \sem{ \AS(\G(D_k\backslash R))}_{-i}(\calA_k)$
            \STATE $Y_k = \sem{\NZ( \F R)}_i(\calA_k)$
            \STATE $D_{k+1} =
                \sem{\AS(\G(Y_k))}_i(\calA_k)
                $
            \STATE $\calA_{k+1} \gets~\calA_{k}$ where
                $\Act_i^{\calA_{k+1}}=\stay_i^{\calA_k}(D_{k+1})$
        \UNTIL{$D_k = D_{k+1}$}
    \end{algorithmic}
\end{restatable}

\begin{restatable}{lemma}{lemterminationas}
    \label{lem:terminationas}
     The almost-sure reachability algorithm of \cite{AHK07}, instantiated
     symbolically on 2.5-player \modelname{} in \cref{alg:asreach}, terminates and returns
     a downward-closed set.
\end{restatable}
To prove the correctness of \cref{alg:asreach},
we must provide a winning
strategy from every state in the computed set $ D_k $, and a strategy for the opponent from every state in the \textit{complement} set $ \overline{D_k} $.
At this point,
it is important to note that positional strategies may not be sufficient, especially
for the opponent player. As illustrated by the \textit{Hide-or-Run game}
presented in \cite{AHK07}, \textit{counting strategies} may be needed.
We generalise the notion of counting strategies to the infinite state case as follows:

\begin{definition}
    \label{def:counting}
    For any $k$, let $p_k=2^{-1/(2^{k})}$.
    A strategy $\sigma_i\in\strat_i$ is \defn{counting} (C) if
    there exist
    two PFM strategies $\sigma_i^v,\sigma_i^u\in\strat_i$
    such that
    for every $k\in\bbN$, $h\cdot s\in S^k$ and any $\alpha \in \Act_i(s)$,
            \[
                \sigma_i(\alpha~|~h\cdot s) =
                p_k \sigma^u_i(\alpha~|~h\cdot s) +
                (1-p_k) \sigma^v_i(\alpha~|~h\cdot s)
            \]
\end{definition}
Note that $p_k$ is a sequence of reals between $0$ and $1$, such
that the infinite product
$\prod_{i=1}^{\infty} p_i=p$ is between $0$ and $1$. This means that the strategy
$\sigma_i^u$ always has some positive probability of being played, but overall
cannot be played forever.
Since the sequence $(p_k)_k$ is fixed, a counting strategy requires infinite memory but can be
finitely represented by two PFM strategies.

\begin{example}
    \label{ex:countreach}
    Consider again the example shown in \Cref{fig:example2}. We provide two
    examples of winning strategies in the $\AS$ case:
    \begin{itemize}
        \item If $\Phi_1=\AS(\F (l_1\cdot \{a\}^3\cup l_2\cdot M^*))$
            --namely Sender can almost-surely
            eventually force a valid transmission of
            $3$ consecutive $a$'s, assuming that the game continues
            forever--
            Sender has a winning
            strategy by playing $\uniform(\{a,b\})$ from all states. For any
            strategy $\sigma_2$ of Attacker, either the game eventually
            reaches $l_2$, or there is a state $s=l\cdot \mu\in \{l_0,l_1\}\cdot M^*$
            visited infinitely often. From this state, there is a fixed
            probability $p>0$ to produce three consecutive $a$'s and that all
            $b$'s are dropped, so this event eventually happens almost-surely.
        \item If $\Phi_1=\AS(\F L\cdot M^*\{a\}M^*)$ --namely a
            message $a$ is eventually sent almost-surely-- we argue that
            Sender cannot achieve her objective: To observe this, one can exhibit a winning
            strategy for Attacker, whose objective is then
            $\Phi_2=\NZ(\G L\cdot \{b,c\}^*)$. Such a strategy consists
            of the counting strategy
            playing action $w$ with probability $p_k$ at round $k$
            and $\uniform(\{a,b\})$ otherwise. At any round still in $l_0$,
            Sender cannot
            risk playing $a$ since there is a small probability for Attacker
            to scramble the communication and then reach $l_2$.
            The overall probability to stay in $l_0\cdot b^*$ is therefore $\prod_{k}p_k > 0$.

            On the other hand, any PFM strategy by Attacker can be defeated, which
            proves that counting strategies are required:
            by playing $b$ from all states in $l_0\cdot b^*$,
            Sender ensures that she will never lose:
            either $\sigma_A$ eventually plays $w$ with probability $1$,
            and Sender can then play $a$ and win, or there is a fixed
            positive probability (FM) that the game moves to $l_1\cdot b^*$, which
            happens almost-surely.
    \end{itemize}
\end{example}

This allows us to conclude on the almost-sure reachability case by adapting
the correctness proof of the almost-sure reachability algorithm in~\cite{AHK07} mentioned in \cref{lem:terminationas}.
As a matter of fact, the finite
attractor property seen in \cref{prop:regclcg} allows us to refer back to the
finite state case and derive sufficient strategies for both players. We
conclude this subsection with the following:
\begin{lemma}
  \label{lem:ascorrectness}
    Using \cref{alg:asreach}, one can compute:
    \begin{itemize}
        \item The winning set $W=\sem{\AS(\F R)}_i$;
    \item A PFM winning strategy for $i$,
         $\sigma_i: h\cdot s\in S^*\cdot W\mapsto \uniform(\stay_i(W))$; and
    \item A
        counting (C) winning strategy for his opponent $-i$ (objective $\NZ(\G \overline{R})$).
   \end{itemize}
\end{lemma}

%% file: stratfig.tex
\begin{figure}
\begin{center}
\begin{tikzpicture}
   \def\topx{0}
   \def\h{1}
   \def\w{10}
   \newcommand{\vertbar}[1]{
      \draw[thick] ($(#1.north east)-(0.05,0)$) -- ($(#1.south east)-(0.05,0)$);
   }
 \begin{scope}[every node/.style={
    rectangle, draw,
    anchor=west,
    execute at begin node={$\vphantom{l_1\mu_1m^l_{k_0}}$}
 }]
   \node (l0) {$l_0$};
   \node at (l0.east) (mu0) {$\mu_0=m^0_1\hdots m^0_{k_0}$};
   \vertbar{mu0}
   \node at (mu0.east) (l1) {$l_1$};
   \node at (l1.east) (mu1) {$\mu_1$};
   \vertbar{mu1}
   \node[right of=mu1, node distance=50] (ln) {$l_n$};
   \node at (ln.east) (mun) {$\mu_n$};

 \end{scope}
 \node at ($(mu1)!0.5!(ln)$) {$\hdots$};
 \path (l0.north west) edge[decorate, decoration = {brace, raise=0.7}]
    node[above, yshift=1.7] {location}
    (l0.north east);
 
 \path (mu0.south east)
     edge[decorate, decoration={brace, raise=0.7}]
     node[below, yshift=-1.7] {Channel state}
    (mu0.south west);
 
 \path (l1.north west)
     edge[decorate, decoration={brace, raise=0.7}]
     node[above, yshift=1.7] {state}
 (mu1.north east);
 \path
    (mun.south east)
    edge[decorate, decoration={brace, raise=0.7}]
    node[below, yshift=-1.7, anchor=north west,pos=1] {Last state (\textbf{P})}
    (ln.south west);
 \path ($(mun.east)+(0.2,0)$) edge[thick,->]
    node[pos=1, anchor=west] {Distribution $\delta$ of actions}
    node[pos=1, anchor=west, yshift=-12] {
        (\textbf{D} if $\delta[\alpha] = 1$)
    }
     +(1,0)
    ;
 \path (mun.south east)
    edge[decorate, decoration={brace, raise=15.7}]
    node[below, yshift=-16.7, anchor=north west, pos=1] {
        History in $S^+=(L\cdot M^*)^+$, \textbf{FM} if finitely many
        regular sets.
   }
 (l0.south west);

 \node[left of=l0] {\huge $\sigma:$};
\end{tikzpicture}
    \caption{Summary of strategy classes.}
    \label{fig:strategies}
\end{center}
\end{figure}

%% file: conjobj.tex
\section{Conjunction of Objectives}
\label{sec:conjobj}
In this section, we stay in the 2.5-player zero-sum setting, and address the
computation of winning regions for objectives
composed as a \textit{conjunction} of qualitative objectives.
%
Note that contrary to the previous section, games are in general
not determined for such conjunctive objectives~\cite{WW21}\todo{cite others?},
so we focus now on the winning strategies for the player whose objective
is a conjunction of objectives.

As discussed in \cref{sec:zerosum}, positive reachability
and almost-sure safety/reachability can be achieved using PFM strategies, while
positive safety may require the use of counting strategies.
The rest of this section is dedicated to proving the following theorem by
combining PFM and counting strategy classes.
\begin{theorem}
    \label{thm:conjunctions}
    Let $\Phi$ be a conjunction of $\NZ$ and $\AS$ objectives
    for safety and reachability path specifications.
    Then the winning region $\sem{\Phi}_i(\calG)$ is computable.
\end{theorem}

We first notice that conjunction of objectives usually
enjoys some convexity properties, namely randomization between individual
optimal strategies is sometimes sufficient to achieve the conjunctive
objective.
This is the case when the conjunction $\Phi$ does not contain a positive
safety objective. Following this observation, we get:
\begin{lemma}
    \label{lem:nonz}
    If $\Phi$ is a conjunctive objective containing
    only $\NZ(\F \cdot),\AS(\G \cdot)$ and $\AS(\F \cdot)$
    objectives of regular sets,
    then $\sem{\Phi}_i(\calG)$
    is computable, and 
    there exists a winning PFM strategy for player~$i$.
\end{lemma}

\begin{proof}[Sketch]
    For every objective $\Psi = \AS(\G R)$ or $\Psi = \AS(\F R)$ appearing in $\Phi$,
    we restrict the set of available actions to $\stay_i(\sem{\Psi}_i)$, as
    a winning strategy necessarily plays actions from this set only.
    States where no actions are available are removed.
    Finally, we solve the game on the restricted arena 
    for each $\NZ(\F\cdot)$ objective, and take the intersection
    of winning regions.
    A sufficient winning strategy consists in playing all actions uniformly at
    random.
\end{proof}

The treatment of positive safety objectives requires more attention, as winning
strategies may require infinite memory, as seen in~\cref{lem:ascorrectness}.
We notice however that any positive safety objective can be replaced by
a simple positive reachability property: it is sufficient for a
winning strategy to wait for all almost-sure reachability
objectives to be met, then play the counting strategy for $\NZ(\G\cdot)$.
It is therefore sufficient and necessary, for at least one state in $\sem{\NZ(\G \cdot)}_i$
to be reachable with positive probability.
This is summarized by the final lemma below:
\begin{lemma}[From $\NZ(\G)$ to $\NZ(\F)$]
    \label{lem:nz}
    If $\Phi\equiv\Theta\wedge \NZ(\G \overline{R})$, then
    $\sem{\Phi}_i = \sem{\Theta \wedge \NZ(\F R')}_i$
    where $ \Theta $ is some conjunctive condition (as defined in \cref{def:qualregions}) and
    \[
        R' =
            \sem{\NZ(\G \overline{R})}_i \cap
            \bigcap_{\AS(\F R'')\in \Phi} R''.
    \]
\end{lemma}

\Cref{thm:conjunctions} is then obtained by applying \Cref{lem:nz} for every
$\NZ(\G \cdot)$ objective, then applying \Cref{lem:nonz}.

%% file: core.tex
\section{Decidability of \ecore and \acore}
\label{sec:core}

In this section, we present our main decidability result. We begin by recalling that \ecore and \acore are undecidable when players' goals are given by almost-sure LTL formulae (\cref{prop:undec}). To obtain decidability, we focus on two types of objectives: almost-sure reachability and almost-sure safety objectives. 
Consider a \modelname\ $ \Game = (\Arena, (\Phi_i)_{i \in \Agt}) $ in which for each $ \Phi_i = \AS(\varphi_i) $, either $ \varphi_i \equiv \F R_i $ or $ \varphi_i \equiv \G R_i $. 
Then the formula used in Step 2 of the procedure for solving \ecore in \cref{sec:core-veri} is given as:
%
%
\begin{align*}
	\Phi_{W} 
			 = \bigwedge_{i\in W} \AS( \varphi_i ) \wedge \bigwedge_{i\notin W} \NZ( \neg \varphi_i  ) \wedge \Gamma.
\end{align*}
%
%
The formula used in Step 3 is given as
%
	$ \Psi_C = \bigwedge_{i \in C} \AS(\varphi_i ) $.
%

Observe that the TPCG $ \Game^{\Agt,\Phi_{W}} $ in Step 2 is, in fact, a
$1.5$-player game, since player $ 1 $ consists of all players in the original game. Thus, solving Step 2 amounts to finding a scheduler that satisfies $ \Phi_{W} $. This problem is decidable if $ \Gamma = \AS(\varphi) $ and $ \varphi $ is of the form $ \bigwedge_i \F R_i $, $ \bigwedge_i \G R_i $, or $ \bigwedge_i \G \F R_i $~\cite{BBS07}. Furthermore, if such a scheduler exists, then there is one that is deterministic and has finite memory. To solve Step 3, we reduce it to solving, for each $ C \subseteq \Agt $, a $2.5$-player zero-sum game in which the goal of player 1 is $ \Psi_C $. As shown in \cref{thm:conjunctions}, this problem is decidable. Therefore, we obtain the following theorem.

\begin{theorem}
	For a pair $ (\Game, \Gamma) $ where players' objectives are almost-sure reachability or almost-sure safety objectives, and property $ \Gamma = \AS(\varphi) $ with $ \varphi $ of the form $ \bigwedge_i \F R_i $, $ \bigwedge_i \G R_i $, or $ \bigwedge_i \G \F R_i $, the problems of \textsc{E-Core} and \textsc{A-Core} are decidable.
\end{theorem}

%% file: conclusion.tex
\section{Concluding Remarks}\label{sec:conclusion}
This paper presents the first result on rational verification of concurrent
stochastic games on infinite arenas, specifically for lossy channel games. Our focus is on the decidability of verifying the core of multi-player concurrent stochastic lossy channel games. To this end, we provide an approach for solving 
concurrent 2.5-player zero-sum games with a conjunction of almost-sure safety and almost-sure reachability 
objectives, and infinite state arenas. Our approach extends previous methods, which were limited to either finite state arenas or turn-based games, to work on concurrent settings and infinite state arenas.

It is worth mentioning that most of the algorithms described in this work have non-elementary complexity in the worst case, as they address problems that can be reduced to the reachability
problem in lossy channel systems~\cite{phs02}.
One might wonder if it is possible to trade stochasticity for non-determinism in
the model to obtain an easier problem. Surprisingly, although the
winning region for non-deterministic reachability (i.e., equivalent
to $\NZ(\F \cdot)$) is computable, the winning region for non-deterministic safety (``does
there exists an infinite path $\hdots$'') is non-computable, due to
undecidability results by~\cite{Mayr03}.

Finally, an obvious direction for future work is to consider the NE concept and its related variants, such as strong NE and coalition-proof NE. However, addressing questions related to NE in concurrent games with probabilistic qualitative objectives is challenging, even in the finite state case. There are two main challenges: Firstly, because NE are strategy profiles where players do not behave as a coalition, limiting their ability to prevent deviations (see \cref{rem:jointdist}), finding the appropriate joint strategies is not straightforward. Secondly, encoding stability against any deviation in the NE setting produces more complex conjunctions of objectives, requiring significant extensions of the techniques presented in \cref{sec:conjobj}.
Reductions from concurrent to turn-based two-players games, have been used in previous work. For instance, in \textit{suspect games}~\cite{BBMU15} one player proposes a NE
while the second attempts to disprove it. Another approach is by
sequentializing games in order to compute \textit{punishing strategies/regions} to characterize NE~\cite{GHW15,GNPW20}.
However, these approaches encounter similar challenges above when stochastic
behaviours are introduced.
We conjecture that in order to capture the NE concept in its generality,
the reduction to 2.5-player games should feature concurrent actions to account for randomized actions and even counting strategies in the NE.

\OMIT{
Future work (not addressed here):
    \begin{itemize}
        \item NE's challenge: combining initial strategy with deviating
            strategies. The usual suspect game/sequentialization do not
            scale to the randomized strategy setting since players proposes
            distributions of actions, not simple actions. One could not simply
            abstract this proposition by its support as we have seen that
            actual probability values matter (counting strategies). Also,
            beyond the $2$-player case, dependence between proposed
            actions may arise (see \cref{rem:jointdist}). These problems already
            appear in the finite state case and are orthogonal to the current
            study.
        \item $\omega$-regular objectives in their full generality.
        \item Might require memory beyond simple counting
        \item Possible extension to the n-player parameterized case
        \item Regular model checking questions are also open (most winning
            regions are also regular).
        \item Limit-sure objectives have not been discussed
        \item Sure Reachability's winning region cannot be computed, however
            it is not clear whether decidability of membership can be computed
            or not. This question is nonetheless 
    \end{itemize}
}

%% file: appendix.tex
We provide here proofs and more detailed explanation from the main article.

\section*{Proof of \cref{prop:regclcg} (Regularity of \modelname )}

We recall the regularity Proposition:
\regclcg

\begin{proof}[Proof ($\cleq$ is wqo)]
    We define \[
        {last}(\mu) = \left\{\begin{aligned}
            m~&~\text{if}~\mu \in M^*\cdot m \\
            \epsilon~&~\text{if}~\mu=\epsilon
        \end{aligned}\right.
        \]
    Let $(l_k\cdot \mu_k)_k = (s_k)_{k\in\bbN}$
    be an infinite sequence of states.
    For every $m\in M \uplus \{\epsilon\}$,
    define:
    $ K_m = \{k\in\bbN~|~{last}(\mu_k) = m\}$.
    For $k \in \bbN$, we have by definition $k\in K_{{last}(\mu_k)}$.
    Then $\bbN = \cup_{m\in M\uplus\{\epsilon\}} K_m$ is a partition of an
    infinite set, hence there exists
    at least one infinite set $K_m$, $m\in M\uplus\{\epsilon\}$.
    Since the sequence $(s_k)_{k\in K_m}$ is infinite and $\preceq$ is a wqo,
    there exist $i\neq j\in K_m$ with $s_i \preceq s_j$ and since
    ${last}(\mu_i)={last}(\mu_j)=m$
    we also have
    $s_i \cleq s_j$.
\end{proof}

\begin{proof}[Proof of 1.]
    Let $s=l\cdot \mu,t=\tilde{l}\cdot \tilde{\mu}$
    two states such that $s\in \pre_i(X)$ and $s\cleq t$.

    Consider $\sigma_i\in\strat_i$ such that
    for all $\vec{\theta}\in\vec{\strat}$,
    $\Pr^{\vec{\theta}_{-i},\sigma_i}_s(\X X)>0$.
    We fix $\vec{\theta}\in\vec{\strat}$, and prove
    $\Pr^{\vec{\theta}_{-i},\sigma_i}_t(\X X)>0$ that is to say that the same
    actions\footnote{In fact, a single deterministic action
    ($\alpha$ s.t. $\sigma_i(\alpha~|~s)>0$) is sufficient
    for player~$i$, but this is not required for the current proof.}
    played from state $t$ are sufficient.

    There exists 
    $(l'\cdot \mu')\in X$
    such that
    $\Pr^{\vec{\theta}_{-i},\sigma_i}_s(\X
        (l'\cdot \mu'))>0$.
    Furthermore, there exist $\alpha$ and $f\in \op_M$ such
    that:
    \begin{itemize}
        \item $\sigma_i(\alpha~|~s) > 0$;
        \item $\Tab(l\cdot \mu)[l',f] > 0$;
        \item $P_{\lambda}(f(\mu), \mu') > 0$.
    \end{itemize}

    $s\cleq t$ so $(l\cdot \mu) \preceq
    (\tilde{l}\cdot \tilde{\mu})= t$ hence $l=\tilde{l}$
    (since the alphabets $L$ and $M$ are disjoint) and
    ${last}(s)={last}(t)$.

    Therefore, the same actions are allowed for all players, hence the same
    strategy can be applied from $t$:
    \begin{itemize}
        \item $\sigma_i(\alpha~|~t) > 0$;
        \item $\Tab(l\cdot \tilde{\mu})[l',f] > 0$;
    \end{itemize}
    
    It remains to prove that $P_{\lambda}(f(\tilde{\mu}), \mu') > 0$,
    or, equivalently, that $\mu'\preceq \tilde{\mu}$.
    We distinguish the three cases:
    \begin{itemize}
        \item If $f={nop}$, then
            $\mu'\preceq f(\mu)=\mu\preceq \tilde{\mu}=f(\tilde{\mu})$;
        \item If $f=!m$, then $\mu'\preceq
            f(\mu)=\mu\cdot m\preceq \tilde{\mu}\cdot m=
            f(\tilde{\mu})$;
        \item If $f=?m$, then $\mu=a_1\cdots a_k\cdot m$ and
            $\tilde{\mu}=w_0\cdot a_1 \cdots a_k \cdot w_k \cdot m \cdot w_{k+1}$
            for some $a_1 \hdots a_k \in M$ and $w_0 \hdots w_{k+1}\in M^*$.
            \begin{itemize}
                \item If $w_{k+1}=\epsilon$,
                    then $f(\tilde{\mu})=w_0\cdot a_1\cdots a_k$;
                \item Otherwise $w_{k+1}=w_{k+1}'\cdot x$, then
                    $x=m$ (because ${last}(\mu) = {last}(\tilde{\mu})$)
                    so $f(\tilde{\mu}) = w_0 \cdot a_1 \cdots a_k \cdot
                        w_k \cdot m \cdot w_{k+1}'$.
            \end{itemize}
            In both cases, $f(\mu) = a_1\cdots a_k\preceq f(\tilde{\mu})$.
    \end{itemize}

    We conclude by proving that any upward-closed set $V$ for $\cleq$ is
    regular: It can be written as a finite union
    $V=\cup_{i=1}^k V_i$ where for all $i$,
    $V_i = \{s\in S~|~v_i \cleq s\}$ for some $v_i \in V$.
    By writing $v_i=l\cdot m_1 \cdots m_k$, we get
    \[
        V_i = \left\{\begin{aligned}
            l&~\text{if}~k=0\\
            l\cdot M^* \cdot m_1 \cdot M^* m_k \cdot M^*&~\text{otherwise}
        \end{aligned}\right.
    \]
\end{proof}

\begin{proof}[Proof of 2.]
    Given a regular language $R$ represented by
    a finite non-deterministic automaton (NFA):

    \begin{itemize}
    \item If $R\subseteq L\cdot M^*$, then one can build a NFA for
        $\uparrow R=\{l\cdot \mu'~|~
        \exists \mu: l\cdot \mu\in R \wedge l\cdot \mu\cleq l\cdot \mu'\}$ by
        duplicating the accepting states (into non-accepting copies) and adding
        a self-loops reading $M$ on every non-accepting states only (this ensures
        that the last letter on the channel is the same,
        according to $\cleq$'s definition).
    \item If $R\subseteq L\cdot M^*$, then for any $l$ one can build a
        NFA for the residual language
        $l^{-1}\cdot R=\{\mu\in M^*~|~l\cdot \mu\in R\}$
        by reading the letter $l$ and updating the initial states accordingly.
    \item If $f\in\op_M$ and $R\subseteq M^*$, one can build a NFA for
        $f^{-1}(R)$;
    \end{itemize}

    Therefore, $\pre_i(R)$ can be written as the following
    sequence of operations:
    \[
        \bigcup_{\alpha}
        \bigcap_{\beta}
        \bigcup_{
            \substack{
            (l_1,l_2,f)\\
            \Tab(l_1,(\alpha,\beta))[l_2,f]>0
            }
        }
        l_1\cdot
        f^{-1}(l_2^{-1}\cdot (\uparrow R))
    \]
    Where
       $\alpha$ and $\beta$ range over the finite set $\Act$ and
       $(l_1,l_2,f)$ ranges over the finite set $L^2\times \op_M$.
\end{proof}

\begin{proof}[Proof of 3.]
    Let $N_L$ denote the random variable counting the number of visits of
    the set of states $L$, that is to say the number of times the channel
    is empty. We will prove that $\esp(N_L)=\infty$ hence
    $\Pr(\text{the channel is empty infinitely often})=1$.

    We do an analysis in the \textbf{worst-case}
    scenario where messages are always
    produced (operation $!m$) and never consumed (operation $?m$), so the
    only way for the channel configuration to decrease in length is
    by message drops.

    For any $k\geq 0$, let $p_k = \Pr^{\sigma}(\X^k L)$ namely the
    probability that the channel is empty at the $k$-th iteration
    \begin{itemize}
        \item From state $0$ to $1$, the produced message can
            be dropped with probability
            $\lambda$ at rounds $0$, $1$, $\hdots$ $k-1$,
            hence a global probability of
            $1-(1-\lambda)^k$.
        \item From state $1$ to $2$, the produced message can
            be dropped with probability
            $\lambda$ at rounds $1$, $2$, $\hdots$ $k-1$,
            hence a global probability of
            $1-(1-\lambda)^{k-1}$.
        \item $\hdots$
        \item From state $k-1$ to $k$, the produced message can
            only be dropped immediately, with
            probability $\lambda=1-(1-\lambda)^1$.
    \end{itemize}
    We conclude that:
    \[
        p_k = \prod_{j=0}^{k-1} (1-(1-\lambda)^{k-j})
            = \prod_{j=1}^{k} (1-(1-\lambda)^{j})
    \]
    We apply Raabe's test for the series $\sum_{k\geq 0} p_k$:
    \[
        \rho_k = k\left(\frac{p_k}{p_{k+1}}-1\right)
        = k\frac{(1-\lambda)^k}{1-(1-\lambda)^k}
    \]
    Since $0<\lambda<1$, $\lim_{k\rightarrow \infty} \rho_k = 0$,
    hence:
    \[
        \esp(N_L)\geq \sum_{k\geq 0}p_k =\infty
    \]

\end{proof}

\input{app-absorbing}

\section*{Positive Reachability Algorithm}
\begin{algorithm}[b]
\caption{Positive Reachability algorithm}
\label{fig:posreachalgo}
    \textbf{Input}: An arena $\calA$, $i\in \Agt$ and a regular set $R$\\
    \textbf{Output}: $U_k = \sem{\NZ(\F R)}_i$
    \begin{algorithmic}
    \STATE $k\gets 0; U_0 \gets R$
    \REPEAT
        \STATE $U_{k+1} \gets U_{k}\cup \pre_i(U_k)$
    \UNTIL{$U_k = U_{k+1}$}
    \end{algorithmic}
\end{algorithm}

We restate first \cref{lem:reachalgo} and explicitly
describe it in~\cref{fig:posreachalgo}:
\lemreachalgo

\begin{proof}
    \begin{itemize}
        \item Every $U_k$ can be be written as $U_k = R \cup V_k$ with
        $V_k$ upward-closed set. Moreover, for all $k$, $V_{k}\subseteq V_{k+1}$
        so by wqo property, there exists $k$ such that
        $V_{k+1}=V_k$ and the algorithm terminates.

    \item Winning strategy for $i$: We
        define $\sigma_i(s)$ for every $s\in S$ by
        \begin{itemize}
        \item If $s \in U_{k+1}\backslash U_k$ for some $k\in\bbN$,
            define $\sigma_i(s)\in \dist(\Act_i(s))$ such that
            for all action $\alpha_{-i}$ of the opponent,
            \[
                \sum_{\alpha_i} 
                \sigma_i(\alpha_i~|~s)
                p(s,(\alpha_i,\alpha_{-i}),U_k)>0
            \]
        \item Define $\sigma_i(s)$ arbitrary otherwise.
    \end{itemize}

    \item To define a spoiling strategy,
    We define $\sigma_{-i}(s)$ for every $s\in S$ by
    \begin{itemize}
        \item If $s\in U_K$, take $\sigma_{-i}(s)$ arbitrary;
        \item If $s\in \overline{U_K}$.
            Since $U_K=\pre_i(U_K,\vec{\gamma})$, by considering
            $\delta_i = \uniform(\gamma_i(s))$, there exists
            some action $\alpha_{-i}\in \gamma_{-i}(s)$ such that
            \[
                \sum_{\alpha_i} \underbrace{\delta_i(\alpha_i)}_{>0}
                    p(s,(\alpha_i,\alpha_{-i}),U_K) = 0
            \]
            Hence, for any action $\alpha_i \in \gamma_i(s)$,
            \[
                p(s,(\alpha_i,\alpha_{-i}, U_k)) = 0
            \]
            We therefore fix $\sigma_{-i}(s) = \alpha_{-i}$.

            We check that for every strategy $\sigma_i$ of $i$,
            $\Pr^{\sigma_i,\sigma_{-i}}(\X \overline{U_K})=1$ then
            by induction on $k$:
            $\Pr^{\sigma_i,\sigma_{-i}}(\wedge_{i\leq k} \X^i \overline{U_K})=1$ then by
            going to the limit
            $\Pr^{\sigma_i,\sigma_{-i}}(\G \overline{U_K})=1$, that
            is to say:
            $\Pr^{\sigma_i,\sigma_{-i}}(\F U_K)=0$.
    \end{itemize}
 \end{itemize}
\end{proof}

\section*{Almost Sure Reachability}

We restate the algorithm from \cite{AHK07} for almost-sure reachability
in \cref{alg:asreach}. Note that
$
\sem{\AS(\G X)}_{-i}=
\overline{\sem{\NZ(\F \overline{X})}_i}$.


\newcounter{oldalgo}
\setcounter{oldalgo}{\value{algorithm}}
\setcounter{algorithm}{0}
\begin{algorithm}[t]
    \caption{Almost-Sure Reachability}
    \label{alg:asreach-restate}
    \textbf{Input}: An arena $\calA$, $i\in\Agt$ and a regular set $R$\\
    \textbf{Output}: $D_k=\sem{\AS(\G \F R)}_i$
    \begin{algorithmic}
        \STATE $k\gets 0;\quad D_0 \gets S; \quad \calA_0 \gets \calA$
        \REPEAT
            \STATE $C_{k} =
                \sem{ \AS(\G(D_k\backslash R))}_{-i}(\calA_k)$
            \STATE $Y_k = \sem{\NZ( \F R)}_i(\calA_k)$
            \STATE $D_{k+1} =
                \sem{\AS(\G(Y_k))}_i(\calA_k)
                $
            \STATE $\calA_{k+1} \gets~\calA_{k}$ where
                $\Act_i^{\calA_{k+1}}=\stay_i^{\calA_k}(D_{k+1})$
        \UNTIL{$D_k = D_{k+1}$}
    \end{algorithmic}
\end{algorithm}

\setcounter{algorithm}{\value{oldalgo}}

\lemterminationas

\begin{proof}
    Notice first that $(D_k)_k$ is a decreasing sequence: at iteration
    $k$, the states from
    $\sem{ \NZ(\F (D_k\backslash Y_k))}_{-i}(\calA_k)$ are removed from
    $D_k$.
    We prove by induction on $k\in\bbN$, that the set
    $\overline{D_k}=S\backslash D_k$ is upward-closed:\\
    Initially, $\overline{D_0}=\emptyset$ is upward-closed. Assume now
    the result holds for some $k\in\bbN$, and consider any $s\preceq t$ with
    $s\in \overline{D_{k+1}}$.
    If $s\in \overline{D_k}$, we have
    $t\in \overline{D_k}\subseteq \overline{D_{k+1}}$
    by induction hypothesis.
    Otherwise, $s\in D_k\backslash D_{k+1}=
    \sem{ \NZ(\F (D_k\backslash Y_k))}_{-i}(\calA_k)$.
    By~\cref{lem:reachalgo}, we distinguish the two following
    cases:
    \begin{itemize}
        \item Either $s$ belongs to an upward-closed set $V$, w.r.t.
            the wqo $\preceq_k$ obtained by restricting
            $\preceq$ to the states in $D_k$. Therefore,
            if $t\in D_k$, then $s\preceq_k t$ and finally $t\in V$.
            Otherwise, $t\notin D_k$ so we already have $t\in \overline{D_{k+1}}$;
        \item Either $s\in D_k\backslash Y_k =
            D_k\backslash \sem{\NZ(\F R)}_i(\calA_k)
                    = \sem{\AS(\G (D_k\backslash R))}_{-i}(\calA_k)$.
             This means there exists $\alpha_{-i}$ such that
             for all $\alpha_{i}$,
             $p(s,(\alpha_i,\alpha_{-i}),(D_k\backslash Y_k))=1$.
             By playing the same action from $t$, player $-i$ ensures
             with positive probability to still
             reach one of those states in $D_k\backslash Y_k$,
             so $t\notin D_{k+1}$.
     \end{itemize}
     $(\overline{D_k})_k$ is an increasing sequence of upward-closed
     sets so it converges to a limit in a finite number of iterations.
\end{proof}

We restate now the correctness result of the algorithm, presented
in \cref{lem:ascorrectness}. Note that the proof is for repeated
reachability meaning the algorithm computes the actual set
$\sem{\AS(\G \F R)}_i$. However, by assuming without loss of generality
that the target set $R$ is absorbing, that is to say that set $R$ cannot
be left, we have $\sem{\AS(\G \F R)}_i = \sem{\AS(\F R)}_i$.
\begin{lemma}
  \label{lem:ascorrectness-recall}
    Given a regular set $R$, the almost-sure reachability algorithm
    from \cite{AHK07}, implemented symbolically in
    in \cref{alg:asreach}, 
    effectively
    computes the set $W=\sem{\AS(\G \F R)}_i$.

     Moreover, $\sigma_i: h\cdot s\in S^*\cdot W\mapsto \uniform(\stay_i(W))$
     is a PFM winning strategy for $i$, while there exists
     a C winning strategy for its opponent (objective $\NZ(\F\G \overline{R})$).
\end{lemma}

\begin{proof}
    For any $s\in D_k$ for a fixed $k\in\bbN$, we denote
    \[ 
        v(s) = \inf_{\sigma_{-i}}\Pr^{\sigma}_s( \F R )
    \]
    Recall that $i$ confines the game to $D_k$ (restricted arena $\calA_k$) so
    almost surely all successors are in $D_k$ and we will
    confine the analysis to states in $D_k$.

    As any $s\in D_k$ satisfies $s\in Y_k = \sem{\NZ(\F R)}_i$,
    we have $v(s) > 0$ and
    we can introduce $\epsilon = \max_{a\in A} v(a) > 0$ (finite set).

    For any strategy $\sigma_{-i}$, and any history $(ht)\in D_k^+$ whose
    last state is $t$, we have the following Markov property:
    \[
        \Pr^{\sigma}_s(\G \overline{R}~|~ht)=
        \Pr^{\sigma(h\cdot)}_t(\G \overline{R})
    \]
    (Under the assumption that $\Pr^{\sigma}_s(ht)>0$)
    Therefore, by total probability:
    \[\begin{aligned}
        \Pr^{\sigma}_s(\G \overline{R}) &= 
            \sum_{h\in (D_k)^*}\sum_{a\in A}
            \Pr^{\sigma}_s(ha) \Pr^{\sigma}_s(\G \overline{R}~|~ha) \\
            &\leq
            \sum_{h\in (D_k)^*}\sum_{a\in A}
            \Pr^{\sigma}_s( ha) (1-v(a))\\
            & \leq
            (1-\epsilon) \Pr^{\sigma}_s(\F A)
       \end{aligned}
    \]
    Then by induction on $n\in\bbN$:
    \[
        1-v(s) \leq (1-\epsilon)^n
    \]
    So $v(s)=1$ for any $s\in D_k$.
    
    Repeated reachability is obtained by applying the Markov property on histories
    ending up
    in $R\cap D_k$.
    This shows that $\sigma_i$ is winning on $\cup_{k\geq 0}D_k$.
    
    We prove now that
    there is a counting winning strategy for $-i$ for the opposite objective
    from any state $s\in \overline{D_k}$ at some iteration $k$.

    We distinguish between the instruction that removed $s$:
    \begin{itemize}
        \item If $s$ got removed during the computation of $Y_k$:
            $s\in D_k\backslash Y_k= \sem{\AS(\G (D_k \backslash R))}_{-i}$
            for some $k\in\bbN$,
            there exists $\alpha_{-i}\in\Act_{-i}(s)$ such that
            for every $\alpha_{i}\in \Act_{i}^{\calG_k}(s)$.

            For every history $h\in S^n$ of length $n$, we define
            $\sigma_{-i}(\alpha_{-i}~|~hs) = p_n$ and
            for every other $\beta\neq \alpha_{-i}$,
            $\sigma_{-i}(\beta~|~hs) = \frac{1-p_n}{|\Act_{-i}(s)|-1}$.

            Against any strategy $\sigma_i$, either player~$i$ plays only
            restricted actions from $\Act_i^{\calA}$ and therefore
            the play stays confined to $D_k\backslash Y_k$
            with probability $\prod_{j\geq 0} p_j$.
            Either player~$i$ plays other actions, in which case
            the game is not confined to $\calA_k$ and with non-zero probability
            reaches another unsafe state from $\calA_l$ ($l<k$).

        \item Otherwise, $s$ was removed by the computation of $D_{k+1}$:
            $s\in Y_k\backslash D_{k+1}$.
            We define in this case $\sigma_{-i}(hs)$ uniformly, for
            all $hs\in S^+$.

            Against any strategy $\sigma_i$, there is a non-zero probability
            to reach in finitely many steps a state $t\in C_l$ for some $l$,
            then a non-zero probability to never reach $R$.
    \end{itemize}
\end{proof}

\section*{Proof(s) of \cref{thm:conjunctions}}

\begin{lemma}[Restate of \cref{lem:nonz}]
    \label{lem:nonz2}
    If $\Phi$ is a conjunctive objective containing
    only $\NZ(\F \cdot),\AS(\G \cdot)$ and $\AS(\F \cdot)$
    objectives of regular sets,
    then $\sem{\Phi}_i(\calG)$
    is computable, and 
    there exists a winning PFM strategy for player~$i$.
\end{lemma}
\begin{proof}
    We define the following regular sets:
    \[
        R_s = \bigcap_{\AS(\G R)\in \Phi} R\quad\text{and}\quad
        R_r = \bigcap_{\AS(\F R)\in \Phi} R
    \]
    \begin{itemize}
        \item We compute $W_s = \sem{\AS(\G R_s)}_i(\calG) = \overline{\sem{\NZ(\F \overline{R_s})}_{-i}(\calG)}$ (\cref{lem:reachalgo});
        \item As $W_s$ is regular, we can build the game $\calG^s$ restricted
            to states in $W_s$ and actions $\stay_i^{\calG}(W_s)$,
            hence for every
            $\vec{\sigma}\in \vec{\strat}^{\calG^s}$, we now have
            $
                \Pr^{\vec{\sigma}}( \G R_s) = 1
            $
        \item We compute $W_r = \sem{\AS(\F R_r)}_i(\calG^s)$ the winning
            region reaching all almost-sure reachability objectives
            simultaneously (\cref{lem:ascorrectness}). Without loss of generality, one can assume
            all reachability objectives to be absorbing;
       \item We further restrict $\calG^s$ to $\calG^r$ by computing
            $\stay_i^{\calG^s}(W_r)$.
       \item Finally, we compute the desired set
       $ W=\cap_{\NZ(\F R)\in\Phi} \sem{\NZ(\F R)}_i(\calG^r)$.
       \item As seen in \cref{lem:ascorrectness},
       any strategy picking all available actions at random is winning for 
        almost-sure reachability objectives. Furthermore, any state $s\in W$
        has a positive probability of satisfying any positive reachability objective.
    \end{itemize}
\end{proof}

Towards a proof of \cref{lem:nonz}, we provide the intermediary lemma
for a necessary condition:
\begin{lemma}
    \label{lem:nzas}
    Let $\vec{\sigma}$ be a strategy profile.
    If 
    $\vec{\sigma},s\vDash \AS(\F R_1) \wedge \NZ(\G \overline{R_2})$,
    then
    $\Pr^{\vec{\sigma}}_{s}(\F R_1 \wedge \G \overline{R_2}) > 0$.
    Moreover, if $\sigma_i$ is winning for
    $\AS(\F R_1) \wedge \NZ(\G \overline{R_2})$,
    then
    $\Pr^{\sigma}_{s}(\F (R_1\cap \sem{\NZ(\G \overline{R_2})}_i))>0$.

\end{lemma}

\begin{proof}
    By law of total probability:
    \[
        \Pr^{\vec{\sigma}}_s(\G \overline{R_2}) = 
        \sum_{hs\in (S\backslash R_1)^* R_1}
            \Pr^{\vec{\sigma}}(\G \overline{R_2} \wedge \cyl(hs)) > 0
    \]
    So there exists an history $hs$ such that
    $0<\Pr^{\sigma}_s(
        \G \overline{R_2} \wedge \cyl(hs)) \leq \Pr^{\sigma}_s(\F R_1 \wedge \G \overline{R_2})
    $

    Assume $\sigma_i$ is winning for
    $\AS(\F R_1)\wedge \NZ(\G \overline{R_2})$ from $s$,
    and by contradiction
    $\Pr^{\vec{\sigma}}_s(\G (\overline{R_1} \cup 
        \sem{\AS(\F R_2)}_{-i}))=1$.
    Since $\sigma_i$ is winning for $\AS(\F R_1)$, we have
    in particular (against $\sigma_{-i}$),
    $\Pr^{\vec{\sigma}}_s(\F R_1) = 1$ so
    $\Pr^{\vec{\sigma}}_s(
        \G\sem{\AS(\F R_2)}_{-i})=1$.
    In particular, $s\in \sem{\AS(\F R_2)}_{-i}$ so $\sigma_{-i}$ is
    not winning for $\NZ(\G \overline{R_1})$ from $s$, hence a contradiction.
\end{proof}

%% file: app-absorbing.tex
\section*{A note about the Absorbing Assumption}
Towards the analysis of almost sure reachability and later the proof of
\cref{thm:conjunctions}, we assume
that any regular reachability/safety objective $B$ is
\defn{absorbing}, that is to say, once reached cannot be exited.
Notice $B$ is a regular set of states, not only locations:
Such a transformation is therefore not possible in the setting
of \modelname{} (as $\Tab$ is defined on a per-location basis as seen
in \cref{def:clcg}). However, this assumption still induces an infinite game
structure whose semantics satisfies the same required properties for
effective computations.

More precisely, we provide the following remark:
\begin{remark}
    \label{lem:absorbing}
    Let $B$ be a regular set of states of $\calA$.
    
    We construct the new set of locations $L_B = L\times \{0,1\}$
    where the second letter/coordinate
    tracks whether $B$ has already been reached, and update the
    definition of $\Pr^{\vec{\sigma}}$ accordingly, namely we enforce that:
    \[
        \Pr^{\vec{\sigma}}(\G (B \rightarrow ((L\times\{1\})\cdot M^*)))=1
    \]
    We then define $B'=(L\times\{1\})\cdot M^*$.
    
    Moreover, in this new game ``arena'' $\calA_B$, we have:
    \begin{itemize}
        \item $\pre_i(X)$ is still upward-closed (for $\cleq$), and can be computed
            whenever $X$ is regular;
        \item $L\times \{0,1\}$ is a finite attractor in $\calA_B$;
        \item The reduction preserves winning
            strategies, as well as FM winning strategies,
            between $\calA$ and $\calA_B$.
    \end{itemize}
    Note however that PFM are not preserved as memory bits might be required
    to remember if some set of states has been reached in the past.
\end{remark}
